\documentclass[letterpaper]{article}
\usepackage[papersize={8.5in,11in},top=1.3in,left=1.3in,right=1.3in,bottom=1.3in]{geometry}
\usepackage[final]{graphicx}
\usepackage{amssymb,amsthm}
\usepackage{amsmath}
\usepackage[boxed]{algorithm2e}
\usepackage{appendix}
\usepackage{bm}
\usepackage{url}
\usepackage{hyperref}
\usepackage{multirow}

\newenvironment{definition}[1][Definition]{\begin{trivlist}
\item[\hskip \labelsep {\bfseries #1}]}{\end{trivlist}}

\newtheorem{theorem}{Theorem}[section]

\newtheorem{lemma}[theorem]{Lemma}
\newtheorem{claim}[theorem]{Claim}
\newtheorem{proposition}[theorem]{Proposition}
\newtheorem{corollary}[theorem]{Corollary}

\newenvironment{thm_app}[1]{\noindent\textbf{Theorem~\ref{#1}.}}{\par\addvspace{\baselineskip}}

\title{Truthful Mechanisms for Matching and Clustering in an Ordinal World}

\author{Elliot Anshelevich \qquad Shreyas Sekar  \\ Rensselaer Polytechnic Institute, Troy, NY 12180\\\texttt{eanshel@cs.rpi.edu, sekars@rpi.edu.}}

\begin{document}
\maketitle


\begin{abstract}
We study truthful mechanisms for matching and related problems in a partial information setting, where the agents' true utilities are hidden, and the algorithm only has access to ordinal preference information. Our model is motivated by the fact that in many settings, agents cannot express the numerical values of their utility for different outcomes, but are still able to rank the outcomes in their order of preference. Specifically, we study problems where the ground truth exists in the form of a weighted graph of agent utilities, but the algorithm can only elicit the agents' private information in the form of a preference ordering for each agent induced by the underlying weights. Against this backdrop, we design truthful algorithms to approximate the true optimum solution with respect to the hidden weights. Our techniques yield universally truthful algorithms for a number of graph problems: a $1.76$-approximation algorithm for Max-Weight Matching, $2$-approximation algorithm for Max $k$-matching, a $6$-approximation algorithm for Densest $k$-subgraph, and a $2$-approximation algorithm for Max Traveling Salesman as long as the hidden weights constitute a metric. We also provide improved approximation algorithms for such problems when the agents are not able to lie about their preferences. Our results are the first non-trivial truthful approximation algorithms for these problems, and indicate that in many situations, we can design robust algorithms even when the agents may lie and only provide ordinal information instead of precise utilities.
\end{abstract}

\section{Introduction}
In recent years, the field of algorithm design has been marked by a steady shift towards newer paradigms that take into the account the behavioral aspects and communication bottlenecks pertaining to self-interested agents. In contrast to traditional algorithms that are assumed to have complete information regarding the inputs, mechanisms that interact with autonomous individuals commonly assume that the input to the algorithm is controlled by the agents themselves.  In this context, a natural constraint that governs the process by which the algorithm elicits inputs from these agents is  \emph{truthfulness}: agents cannot improve upon the resulting outcome by misreporting the inputs. Another constraint that has recently gained traction in optimization problems on weighted graphs (where the agents correspond to the nodes) is that of \emph{ordinality}: here, each agent can only submit a preference list of their neighbors ranked in the order of the edge weights. The need for algorithms that are both truthful and ordinal arises in a number of important settings; however, it is well known that it is impossible to obtain optimum solutions even when the algorithm is required to satisfy only one of these two constraints.

In this work, we study the design of approximation algorithms for popular graph optimization problems including matching, clustering, and team formation with the goal of understanding the \emph{combined price of truthfulness and ordinality}. To be more specific, we consider the above optimization problems on a weighted graph whose vertices represent the agents, and where the edge weights (that correspond to agent utilities) are private to the agents constituting that edge, and pose the following natural question: ``\emph{How does a computationally efficient, truthful algorithm that only has access to each agent's edge weights in the form of preference rankings perform in comparison to an optimal algorithm that has full knowledge of the weighted graph?}".

\paragraph{\textbf{Truthfulness in an ordinal world}} Mechanisms that are either truthful or ordinal have received extensive attention across the spectrum of optimization problems. However, non-trivial algorithms that satisfy both of these considerations exist only for very specific settings~\cite{feldman16, amanatidisBM16}. For instance, the price of ordinality (also referred to as \emph{distortion}) is well understood for a number of applications such as voting~\cite{anshelevichBP15,boutilierCHLPS15}, matching~\cite{filosRatsikasF014,anshelevichS16}, facility location~\cite{feldman16}, and subset selection~\cite{anshelevichS16,regret}. The common thread in all of these settings where the (input) information is often held by the users is that it may be impossible or prohibitively expensive for the agents to express their full utilities to the mechanism; the same agents may incur a smaller overhead if they communicate preference lists over the other users or candidates in the system. Our main contention in this paper is that in exactly the same types of settings, it is reasonable to expect strategic agents to lie about their preferences if it improves their resulting utilities. Motivated by this, we study ordinal algorithms that are also truthful. Even though such mechanisms are clearly less powerful than their `ordinal but not necessarily truthful' counterparts, our high level-level contribution is that for several well-studied graph maximization problems, one can obtain solutions that are only a constant factor away from the (social welfare of the) optimum, omniscient solution.

\paragraph{Model and Problem Statements}
The high-level model in this paper is the same as the one in \cite{anshelevichS16}, with the addition of truthfulness as a constraint. The common setting for all the problems studied in this work is an undirected, complete weighted graph $G$ whose nodes are the set of self-interested agents $\mathcal{N}$ with $|\mathcal{N}|=N$. We use $w(x,y)$ to denote the weight of the edge $(x,y)$ in the graph for $x,y \in \mathcal{N}$. All of the optimization problems studied in this work involve selecting a subset of edges from $G$ that obey some condition, with the objective of maximizing the weight of the edges chosen.
\begin{description}

\item [Max $k$-Matching] Compute the maximum weight matching consisting of exactly $k$ edges. We refer to the $k=\frac{N}{2}$ case as the \emph{Weighted Perfect Matching} problem.

\item [$k$-Sum Clustering] Given an integer $k$, partition the nodes into $k$ disjoint sets $(S_1, \ldots, S_k)$ of equal size in order to maximize $\sum_{i=1}^k \sum_{x,y \in S_i}w(x,y)$. (It is assumed that $N$ is divisible by $k$). When $k=N/2$, $k$-sum clustering reduces to the weighted perfect matching problem.

\item [Densest $k$-subgraph] Given an integer $k$, compute a set $S \subseteq \mathcal{N}$ of size $k$ to maximize the weight of the edges inside $S$.

\item [Max TSP] In the maximum traveling salesman problem, the objective is to compute a tour $T$ (cycle that visits each node in $\mathcal{N}$ exactly once) to maximize $\sum_{(x,y) \in T}w(x,y)$.
\end{description}

A crucial but reasonably natural assumption that we make in this work is that the edge weights satisfy the \emph{triangle inequality}, i.e., for $x,y,z \in \mathcal{N}$, $w(x,y) \leq w(x,z) + w(z,y)$. For the specific kind of the problems that we study, the metric structure occurs in a number of well-motivated environments such as: $(i)$ social networks, where the property captures a specific notion of friendship, $(ii)$ Euclidean metrics: each agent is a point in a metric space which denotes her skills or beliefs, and $(iii)$ edit distances: each agent could be represented by a string over a finite alphabet (for e.g., a gene sequence) and the graph weights represent the edit or Levenshtein distances~\cite{troncosoKC07}. The reader is asked to refer to~\cite{anshelevichS16} for additional details on these specific applications and a mathematical treatment of friendship in social networks.

Our framework and problem set models a multitude of interesting applications, and not surprisingly, all of the problems described above (with the metric assumption) have been the subject of a dense body of algorithmic work~\cite{anshelevichS16,feo1990class,hassin1997approximation,kowalik2009deterministic}. In many of these applications, it becomes imperative that the algorithm provide good approximation guarantees even in the absence of precise numerical information regarding the graph weights. For instance, one can imagine partitioning a set of wedding guests to form a table assignment ($k$-sum clustering) or selecting a diverse team of agents in order to tackle a complex task (dense subgraph).

\subsubsection*{Algorithmic Framework} In this work, we are interested in the design of algorithms that are both ordinal and truthful. Suppose that for any one of the above problems, we are given an instance described by a weighted graph; then an algorithm $\mathcal{A}$ for this problem is said to be ordinal if it has access only to a vector of preference orderings induced by the graph weights. That is, the input to this algorithm consists of a set of $N$ preference orderings reported by each of the agents, where the preference list corresponding to agent $i \in \mathcal{I}$ is a ranking over the agents in $\mathcal{N} - \{i\}$ such that $\forall j,k \in \mathcal{N}$, if $i$ prefers $j$ to $k$, then $w(i,j) \geq w(i,k)$.

The algorithm is truthful if no single agent can improve their utility by submitting a preference ordering different from the `true ranking' induced by the graph weights. Here, the utility of each agent $i$ is simply the total weight of the edges incident to $i$ which are chosen. These utilities have a natural interpretation with respect to the problems considered in this work. For instance, for matching problems, an agent's utility corresponds to her affinity or weight to the agent to whom she is matched, and for densest subgraph as well as $k$-sum clustering, the utility is her aggregate weight to the agents in the same team or cluster. Our objective in this paper is to design mechanisms that maximize the overall social welfare, i.e., the sum of the utilities of all the agents. Thus, the goal is to select a maximum-weight set of edges while knowing only ordinal preferences (instead of the true weights $w$), with even the ordinal preferences possibly being misrepresented by the self-interested agents.

Finally, $\mathcal{A}$ is said to be an ordinal $\alpha$-approximation algorithm for $\alpha \geq 1$ if for any given instance along with the graph weights,  the total objective value of the maximum weight solution with respect to the instance weights is at most a factor $\alpha$ times the value of the solution returned by $\mathcal{A}$, when the input corresponds to the preference rankings induced by the weights. In other words, such algorithms produce solutions which are always a factor $\alpha$ away from optimum, {\em without actually knowing} what the weights $w$ are. We conclude by pointing out that despite the extensive body of work on all of the problems described previously, hardly any of the proposed mechanisms satisfy either truthfulness or ordinality (see Related Work for exceptions), motivating the need for a new line of algorithmic thinking.

\subsection*{Our Contributions}

Our main results are summarized in Table \ref{table_results}. All of the non-matching problems that we study are NP-Hard even in the full information setting~\cite{feo1990class,ravi1994heuristic,kosarajuPS94}. Our truthful ordinal algorithms provide constant approximation factors for a variety of problems in this setting, showing that even if only ordinal information is presented to the algorithm, and even if the agents can lie about their preferences, we can still form solutions efficiently with close to optimal utility. Note that as seen in Table \ref{table_results}, in \cite{anshelevichS16} the authors already gave ordinal approximation algorithms for matching problems: those algorithms were not truthful, however, and achieving non-trivial approximation bounds while always giving players incentive to tell the truth requires significant additional work. For example, even the natural, greedy 2-approximation algorithm for Max $k$-matching from \cite{anshelevichS16} is {\em not} truthful.

\begin{table}[tb]
\centering
\begin{tabular}{|l|c|c|}
\hline
\textbf{Problem} &  \multicolumn{2}{|c|}{\textbf{Our Results}}\\

& Truthful Ordinal & Non-Truthful Ordinal\\
\hline
Weighted Perfect Matching  & $1.7638$ & $1.6$ \cite{anshelevichS16}\\
\hline
Max $k$-Matching & $2$ & $2$ \cite{anshelevichS16}\\
\hline
$k$-Sum Clustering& $2$ & $2$ \\
\hline
Densest $k$-Subgraph & $6$ & $(\frac{4}{\beta^2},\beta)$ (*) \\
\hline
Max TSP & $2$ & $1.88$\\
\hline
\end{tabular}
\caption{Approximation factors provided in this paper by both truthful and non-truthful ordinal algorithms. (*) A bicriteria result for Densest $k$-subgraph where the set size is relaxed to $\beta k$ but the approximation factor is improved from $4$ to $\frac{4}{\beta^2}$ for $\beta \geq 1$.
} 
\label{table_results}
\end{table}


In addition to considering truthful mechanisms, we also develop new approximation algorithms for the setting where the agents are not able to lie, and thus the algorithm knows their true preference ordering. By dropping the truthfulness constraint, we are able to obtain better approximation factors for clustering, densest subgraph, and max TSP. The improved results are enabled by more involved algorithmic techniques that invariably sacrifice truthfulness; they establish a clear separation between the performance of an unconstrained ordinal algorithm and one that is required to be truthful. 


\noindent\textbf{Techniques:} Our proof techniques involve carefully stitching together \emph{greedy}, \emph{random}, and \emph{serial dictatorship} based solutions. Understandably, and perhaps unavoidably for ordinal settings, the algorithmic paradigms that form the bedrock for our mechanisms are rather simple. However, beating the guarantees obtained by a naive application of these techniques involves a more intricate understanding of the interplay between the various approaches. For instance, our algorithm for the weighted perfect matching problem involves mixing between two simple $2$-approximation algorithms (greedy, random) to achieve a $1.76$-guarantee: towards this end, we establish new tradeoffs between greedy and random matchings showing that when one is far away from the optimum solution, the other one must provably be close to optimum.



\subsection{Related Work}

Algorithms proposed in the vast matching literature usually belong to one of two classes:  $(i)$ Ordinal algorithms that ignore agent utilities, and focus on (unquantifiable) axiomatic properties such as stability, truthfulness, or other notions of efficiency, and $(ii)$ Optimization algorithms where the numerical utilities are fully specified. Algorithms belonging to the former class usually do not result in good approximations for the hidden optimum utilities, while techniques used in the latter tend to heavily rely on the knowledge of the exact edge weights and are not suitable for this setting. A notable exception to the above dichotomy is the class of optimization problems studying \emph{ordinal measures of efficiency}~\cite{abrahamIKM07,chakrabartyS14,bhalgatCK11a,krysta14}, for example, the average rank of an agent's partner in the matching.  Such settings usually involve the definition of `new utility functions' based on given preferences, and thus are fundamentally different from our model where preexisting cardinal utilities give rise to ordinal preferences.



Broadly speaking, the truthful mechanisms in our work fall under the umbrella of `mechanism design without money'~\cite{amanatidisBM16,filosRatsikas16, dughmiG10, filosRatsikasF014, procacciaT13}, a recent line of work on designing strategyproof mechanisms for settings like ours, where monetary transfers are irrelevant. A majority of the papers in this domain deal with mechanisms that elicit agent utilities, specifically for one-sided matchings, assignments and facility location problems that are somewhat different from the graph problems we are interested in. The notable exceptions are the recent papers on truthful, ordinal mechanisms for one-sided matchings~\cite{filosRatsikasF014,filosRatsikas16} and general allocation problems~\cite{amanatidisBM16}. While~\cite{filosRatsikasF014} looks at normalized agent utilities and shows that no ordinal algorithm can provide an approximation factor better than $\Theta({\sqrt{N}})$,~\cite{filosRatsikas16} considers {\em minimum} cost metric matching under a resource augmentation framework. The main differences between our work and these two papers are (1) we consider two-sided matching instead of one-sided, as well as other clustering problems, as well as non-truthful algorithms with better approximation factors, and (2) we consider maximization objectives in which users attempt to maximize their utility instead of minimize their cost. The latter may seem like a small difference, but it completely changes the nature of these problems, allowing us to create many different truthful mechanisms and achieve {\em constant-factor approximations}. Finally, ~\cite{amanatidisBM16} looks at the problem of allocating goods to buyers in a `fair fashion'. In that paper, the focus is on maximizing a popular non-linear objective known as the \emph{maximin share}, which is incompatible with our objective of social welfare maximization. That said, an interesting direction is to see if our techniques extend to other objectives.




As discussed in the Introduction, this paper improves on several results from \cite{anshelevichS16}. In \cite{anshelevichS16}, the authors focused on the problem of maximum-weight matching for the non-truthful setting, with the main result being an ordinal 1.6-approximation algorithm. In the current paper, we greatly extend the techniques from \cite{anshelevichS16} so that they may be applied to other problems in addition to matching. Moreover, we introduce several new techniques for this setting in order to create {\em truthful} algorithms; such algorithms require a somewhat different approach and make much more sense for many of the settings that we are interested in. Other than \cite{filosRatsikas16}, these are the first known truthful algorithms for matching and clustering with metric utilities.



Our work is similar in motivation to the growing body of research studying settings where the voter preferences are induced by a set of hidden utilities~\cite{anshelevichBP15,boutilierCHLPS15,caragiannisP11,randomized,regret,feldman16}. The voting protocols in these papers are essentially ordinal approximation algorithms, albeit for a very specific problem of selecting the utility-maximizing candidate from a set of alternatives.



\section{Preliminaries}
\label{sec:prelim}

\subsection{Truthful Ordinal Mechanisms}

As mentioned previously, we are interested in designing incentive-compatible mechanisms that elicit ordinal preference information from the users, i.e., mechanisms where agents are incentivized to truthfully report their preferences in order to maximize their utility. We now formally define the notions of truthfulness pertinent to our setting. Throughout the rest of this paper, we will use $P_i$ to represent a true ordinal preference of agent $i$ (i.e., one that is induced by the utilities $u(i,j)$), and $s_i$ to represent the preference ordering that agent $i$ submits to the mechanisms (which will be equal to $P_i$ if $i$ tells the truth).

\begin{definition}{(Truthful Mechanism)}
A deterministic mechanism $\mathcal{M}$ is said to be truthful if for every $i \in \mathcal{N}$, all $\vec{s}_{-i}, s'_i$, we have that $u_i(P_i, \vec{s}_{-i}) \geq u_i(s'_i, \vec{s}_{-i})$, where $u_i$ is the utility guaranteed to agent $i$ by the mechanism.
\end{definition}

\begin{definition}{(Universally Truthful Mechanisms)}
A randomized mechanism is said to be universally truthful if it is a probability distribution over truthful deterministic mechanisms.
\end{definition}
Informally, in a universally truthful mechanism, a user is incentivized to be truthful even when she knows the exact realization of the random variables involved in determining the mechanism.

\begin{definition}{(Truthful in Expectation)}
A randomized mechanism is said to be truthful in expectation if an agent always maximizes her expected utility by truthfully reporting her preference ranking. The expectation is taken over the different outcomes of the mechanism.
\end{definition}

All of our algorithms are {\em universally} truthful, not just in expectation.  The reader is asked to refer to~\cite{dobzinskiNS12} for a useful discussion on the types of randomized mechanisms, and settings where universally truthful mechanisms are strongly preferred as opposed to the mechanisms that only guarantee truthfulness in expectation.



\subsection{Approaches for Designing Truthful Matching Mechanisms}
\label{subsec:highlevelapproach}
As a concrete first step towards designing truthful ordinal mechanisms, we introduce three high-level algorithmic paradigms that will form the backbone of all the results in this work. These paradigms are based on the popular algorithmic notions of \emph{Greedy}, \emph{Serial Dictatorship}, and \emph{Uniformly Random}. For each of these paradigms, we develop approaches towards designing truthful mechanisms for the maximum matching problem. 
In Sections~\ref{sec:truthfulmatching} and~\ref{sec:truthfulother}, we develop more sophisticated truthful mechanisms that build upon the simple paradigms presented here, leading to improved approximation factors.

\subsubsection*{Greedy via Undominated Edges:}
Our first algorithm is the ordinal analogue of the classic greedy matching algorithm, that has been extensively applied across the matching literature. In order to better understand this algorithm, we first define the notion of an \textit{undominated edge}.
\begin{definition}{(Undominated Edge)}
Given a set $E$ of edges, $(x,y) \in E$ is said to be an undominated edge if for all $(x,a)$ and $(y,b)$ in $E$, $w(x,y) \geq w(x,a)$ and $w(x,y) \geq w(y,b)$.
\end{definition}

We make two simple observations here regarding undominated edges based on which we define Algorithm~\ref{alg_greedy}.
\begin{enumerate}
\item Every edge set $E$ has at least one undominated edge. In particular, any maximum weight edge in $E$ is obviously an undominated edge.
\item Given an edge set $E$, one can efficiently find at least one undominated edge {\em using only the ordinal preference information}~\cite{anshelevichS16}.
\end{enumerate}

\begin{algorithm}[htbp]
{$M:= \emptyset$, $T$ is the valid set of edges initialized to the complete graph on $\mathcal{N}$\;}
\While{$T$ is not empty}{
pick an undominated edge $e=(x,y)$ from $T$ and add it to $M$\;
remove all edges containing $x$ or $y$ from $T$; if $|M| = k$, $T = \emptyset$.
}
\caption{Greedy Algorithm for Max $k$-Matching}
\label{alg_greedy}
\end{algorithm}

It is not difficult to see that this algorithm gives a 2-approximation for Max-Weight Perfect Matching, and is truthful for that case. Unfortunately, for Max $k$-Matching with smaller $k$, it is no longer truthful, and thus none of the algorithms that use Greedy as a subroutine (such as the algorithms from \cite{anshelevichS16}) are truthful.

\begin{proposition}
\label{clm_truthfulundom}
Algorithm~\ref{alg_greedy} is truthful for the Max $k$-Matching problem only when $k=\frac{N}{2}$.
\end{proposition}

\begin{proof}
We need to prove that for any given strategy profile adopted by the other players $\vec{s}_{-i}$, player $i$ maximizes her utility when she is truthful, i.e., if $P_i$ is the true preference ordering of agent $i$ and $s_{-i}$ is any set of preference orderings for the other agents, then $u_i(P_i, \vec{s}_{-i}) \geq u_i(s'_i, \vec{s}_{-i})$ for any $s'_i$. Our proof will proceed via contradiction and will make use of the following fundamental property: \emph{if Algorithm~\ref{alg_greedy} (for some input) matches agent $i$ to $j$ during some iteration, then both $i$ and $j$ prefer each other to every other agent that is unmatched during the same round}.

We introduce some notation: suppose that $M$ denotes the matching output by Algorithm~\ref{alg_greedy} for input $(P_i, \vec{s}_{-i})$, and for every $x \in \mathcal{N}$, $m(x)$ is the agent to whom $x$ is matched to under $M$. Let $e_j$ be the edge added to the matching $M$ in round $j$ of Algorithm \ref{alg_greedy}, denote the round in which $i$ is matched to $m(i)$ as round $k$. Assume to the contrary that for input $(s'_i, \vec{s}_{-i})$, $i$ is matched to an agent she prefers more than $m(i)$. Let the altered matching be referred to as $M'$, and let $m'(x)$ be the agent who $x$ is matched with in $M'$.

We begin by proving the following claim: {\em For each $j < k$, we have that $e_j\in M'$}. In other words, all the edges which are included into $M$ before $i$ is matched by Algorithm \ref{alg_greedy} must appear in both matchings no matter what $i$ does. Once we prove this claim, we are done, since $e_k$ is the highest-weight edge from $i$ to any node not in $e_1,\ldots,e_{k-1}$, so $i$ maximizes its utility by telling the truth and receiving utility equal to the weight of $e_k$.

To prove the claim above, we proceed by induction. Note that if $k=1$, then $i$ is trivially truthful, since $m(i)$ is its top choice in the entire graph. Now suppose that we have shown the claim for edges $e_1,\ldots,e_{j-1}$. Let $e_j=(x,y)$, and without loss of generality suppose that $x$ is matched in our algorithm constructing $M'$ before $y$. At the time that $x$ is matched with $m'(x)$, it must be that $m'(x)$ is the top choice of $x$ from all available nodes. But, by the definition of our algorithm, $y$ is the top choice of $x$ that is not contained in $e_1,\ldots,e_{j-1}$. Since $m'(x)$ is not contained in $e_1,\ldots,e_{j-1}$ due to our inductive hypothesis, this means that $x$ prefers $y$ over $m'(x)$, and since $y$ is not matched yet, this means that $x$ and $y$ will become matched together in $M'$. Thus, $e_j$ is in $M'$ as well. This completes the proof of our claim.

To see why this mechanism is not truthful for smaller $k$, notice that agents which would not be matched in the first $k$ steps have incentive to lie and form undominated edges where none exist, all in order to be matched earlier. Assume that the algorithm uses a deterministic tie-breaking rule to choose between multiple undominated edges in each round. While this does not really alter the final output for the perfect matching problem, the tie-breaking rule may lead to certain undominated edges not getting selected for the final matching.

Fix $k$ and suppose that when the input preferences are truthful, agents $i$, $j$ are not present in the matching $M$ returned by Algorithm~\ref{alg_greedy}. Moreover, suppose that (1) $j$'s first preference is $i$, and (2) the deterministic tie-breaking always prefers $(i,j)$ over other edges (one can design preferences so that agents favoured by the tie-breaking are not selected for truthful inputs).

Clearly $i$ has incentive to alter its preferences to identify $j$ as its most preferred node and receive a utility of $w(i,j)$, which is more than its previous utility of zero.
\end{proof}

Can we use a similar approach to design algorithms for the other problems that we are interested in? For $k$-sum clustering and Densest $k$-subgraph, one can follow the approach taken in~\cite{hassin1997approximation,anshelevichS16}, and use the above matching as an intermediate to compute $4$-approximations for the above problems. For Max TSP, we can directly leverage the above algorithm  by maintaining $M$ as a (forest of) path(s) instead of a matching in order to obtain a $2$-approximate Hamiltonian tour. Unfortunately, as we show in the Appendix, these approaches {\bf do not} lead to truthful algorithms at all.

\subsubsection*{Serial Dictatorship}
Another popular approach to compute incentive compatible matchings (albeit usually for one-sided matchings~\cite{filosRatsikas16, filosRatsikasF014}) is serial dictatorship, which we formally define below for our two-sided matching setting.

\begin{proposition}
Algorithm~\ref{alg_sd} is universally truthful for the Max $k$-Matching problem for all $k$.
\end{proposition}

\begin{algorithm}[htbp]
{$M:= \emptyset$, $T$ is the set of available agents initialized to $\mathcal{N}$}\;
\While{$T$ is not empty}{
pick an available agent $x$ arbitrarily from $T$\;
let $y$ denote $x$'s most preferred agent in $T - \{x\}$; add $(x,y)$ to $M$\;
remove all edges containing $x$ or $y$ from $T$; if $|M| = k$, $T = \emptyset$.
}
\caption{Serial Dictatorship for Max $k$-Matching}
\label{alg_sd}
\end{algorithm}

Serial dictatorship is among the most prominent of algorithms to feature in this work: our primary approximation algorithms for Max $k$-matching and Max TSP involve randomized versions of serial dictatorship. 


\textbf{Randomness}
A much simpler approach that is completely oblivious to the input preferences involves selecting a solution uniformly at random. Such an algorithm (described in Algorithm~\ref{alg_random}) is obviously truthful. Many of the techniques in this paper rely on carefully combining these three types of algorithms in order to produce good approximation factors while retaining truthfulness.

\begin{algorithm}[htbp]
{$M:= \emptyset$, $T$ is the valid set of edges initialized to the complete graph on $\mathcal{N}$\;}
\While{$T$ is not empty}{
pick an edge $e=(x,y)$ from $T$ uniformly at random and add it to $M$\;
remove all edges containing $x$ or $y$ from $T$; if $|M| = k$, $T = \emptyset$.
}
\caption{Random Algorithm for Max $k$-matching}
\label{alg_random}
\end{algorithm}

\begin{proposition}
\label{prop_randomtruth}
Algorithm~\ref{alg_random} is universally truthful for the Max $k$-matching problem for all $k$.
\end{proposition}

\section{Truthful Mechanisms for Matching}
\label{sec:truthfulmatching}

\subsection{Weighted Perfect Matching}
So far, we have looked at two simply approaches for designing truthful mechanisms (Greedy and Random) for the weighted perfect matching problem, both of which yield $2$-approximations~\cite{anshelevichS16} to the optimum matching. Can we do any better? In~\cite{anshelevichS16}, the authors use a complex interleaving of greedy and random approaches to extract a {\em non-truthful} $1.6$-approximation algorithm. In this paper, we instead present a simpler algorithm and rather surprising result: a simple random combination of Algorithms~\ref{alg_greedy} and~\ref{alg_random} results in a $1.76$-approximation to the optimum matching. The main insight driving this result is the fact that the random and greedy approaches are in some senses complementary to each other, i.e., on instances where the approximation guarantee for the greedy algorithm is close to $2$, the random algorithm performs much better.
\begin{theorem}\label{thm:truthfulMatching}
The following algorihm is a universally truthful mechanism for the weighted perfect matching problem that obtains a $1.7638$-approximation to the optimum matching.

\noindent\textbf{Greedy-Random Mix Algorithm for Weighted Perfect Matching}: With probability $\frac{3}{7}$, return the output of Algorithm~\ref{alg_greedy} for $k=\frac{N}{2}$ and with probability $\frac{4}{7}$, return the output of Algorithm~\ref{alg_random} for $k=\frac{N}{2}$.
\end{theorem}

\begin{proof}\subsubsection*{Notation}
Our proof mainly involves non-trivial lower bounds on the performance of the random matching which highlights its complementary nature to the greedy matching. As usual, we begin with notation that allows us to divide the greedy matching into several parts for easy analysis.

\textbf{Dividing Greedy into Two Halves} Suppose that $GR$ is the output of the greedy algorithm for the given instance, and $RD$ is the random matching for the same instance. We abuse notation and define $T:= \mathcal{N}(GR_\frac{1}{2})$, and $GR(T) = GR_\frac{1}{2}$. Recall that $GR_\frac{1}{2}$ comprises of the top (max-weight) fifty percent of the edges in $GR$. We will some times refer to $T$ as the top half and the rest of the nodes as the bottom half. Next, define $B := \mathcal{N} \setminus T$, and let $GR(B)$ denote $GR \setminus GR(T)$. Observe that both $T$ and $B$ consist of exactly $\frac{N}{2}$ nodes. Finally, suppose that $w(GR(B)) = xOPT$.

\noindent\textbf{Sub-Dividing B} We will now go one step further and divide the bottom half $B$ into two sub-parts, $B_1$ and $B_2$, which will aid us in our analysis of the random matching. Define $GR(B_1)$ to be the top $\frac{xN}{2}$ edges from $GR(B)$, i.e., $GR(B_1) := \{GR(B)\}_{2x}$ since $GR(B)$ consists only of $\frac{N}{4}$ edges. Finally, $GR(B_2)$ is the final part of the greedy matching, i.e., $GR(B_2) = GR(B) \setminus GR(B_1)$. As with our previous definitions, $B_1$ and $B_2$ will represent the nodes contained in $GR(B_1)$ and $GR(B_2)$ respectively.

We begin by highlighting some easy observations in order to get familiar with the various sub-matchings defined above.

\begin{proposition}
\label{prop_simpleproprdgrper}
\begin{enumerate}
\item $w(GR(T)) \geq \frac{1}{2}OPT$.

\item $B_1$ consists of $xN$ nodes and $B_2$ consists of $(\frac{1}{2} - x)N$ nodes.

\item No edge in $GR(B_2)$ can have a weight larger than $\frac{2GR(B_1)}{xN}$.
\end{enumerate}

\end{proposition}
The first part of the Proposition comes from Lemma~\ref{lem_greedy_tophalf}. The last part is simply because this is the average of edge weights in $GR(B_1)$.

The rest of the proof involves proving new lower bounds on the weight of the random matching as a function of $x$. Specifically, we will fix the performance of the greedy matching (fix $x$) and then show that when $x$ is small, the random matching's weight is close to $\frac{5}{8}OPT$. The reminder of the proof is just basic algebra to bring out the worst-case performance. Let us first formally state our trivial lower bound on the greedy matching.

\begin{proposition}
\label{prop_grlowerboundfinal}
The weight of the greedy matching is given by:

$$w(GR) = w(GR(T)) + w(GR(B)) \geq \frac{1}{2}OPT + xOPT.$$
\end{proposition}

Before developing the machinery towards our lower bound for the random matching, we will first state our end-goal, which we will prove later. Essentially, our main claim provides an unconditional lower bound for the performance of the random matching as a well as a (conditional) bound for small $x$, which will serve as the worst-case.

\begin{claim}
\label{clm_mainlbrandommatch}
The weight of the random matching is always at least

$$E[w(RD)] \geq \frac{5}{8}OPT - x(1-\frac{3}{2}x)OPT.$$

Moreover, when $x \leq \frac{1}{8}$, the following is a tighter lower bound for the random matching

$$E[w(RD)] \geq \frac{5}{8}OPT - x(1-2x)OPT.$$
\end{claim}

\subsection*{Tackling the Random Matching for Different Cases}
We will now prove three lemmas that will act as the main bridges to showing Claim~\ref{clm_mainlbrandommatch}. These lemmas provide insight on the random matching for different cases depending on the relative weights of $GR(B_1)$ and $GR(B_2)$. First, define $\alpha := \frac{w(GR(B_1))}{w(GR(B))}$. We begin by studying the case when $\alpha$ is smaller than $\frac{1}{2}$, i.e., the weights of the edges in $GR(B)$ are somewhat evenly distributed across $GR(B_1)$ and $GR(B_2)$. Moreover, since every edge in $GR(B_1)$ is larger than every edge in $GR(B_2)$, the following is an easy lower bound on $\alpha$.

\begin{lemma}
For any given instance where $w(GR(B)) = xOPT$, we have that $\alpha \geq 2x$.
\end{lemma}
\begin{proof}
$GR(B_1)$ consists of $\frac{xN}{2}$ edges whereas $GR(B)$ consists of $\frac{N}{4}$ edges. $\hfill\qed$
\end{proof}

Therefore, the above lemma indicates that when $\alpha \leq \frac{1}{2}$, $x$ canot be larger than $\frac{1}{4}.$

Now we give the first of the three lemmas.

\begin{lemma}
\label{lem_rd_lowerbound1}
Suppose that for a given instance with $w(GR(B)) = xOPT$, $w(GR(B_1)) = \alpha xOPT$ with $\alpha \leq \frac{1}{2}$. Then, for $x \leq \frac{1}{4}$, we have that

$$w(RD) \geq \frac{5}{8}OPT - x(1-2x)OPT.$$

\end{lemma}
\begin{proof}
From Lemma~\ref{lem_rand_lowerbound2}, we get the following generic lower bound for $RD$ since $|B| = \frac{N}{2}$,

$$w(RD) \geq \frac{1}{2}OPT + \frac{1}{N}\{w(T)  - w(B)\}.$$

Moreover, applying Lemma~\ref{lem_matchingupper} to $T$, we also get that $\frac{w(T)}{N} \geq \frac{1}{8}OPT$ since there exists a matching ($GR(T)$) solely on the nodes inside $T$ having a weight of $\frac{OPT}{2}$. Therefore, it suffices to prove an upper bound on $\frac{w(B)}{N}$. Recall that $B$ consists of exactly $n=\frac{N}{2}$ nodes, $GR(B) = xOPT$, and $w(\{GR(B)\}_{2x}) = \alpha xOPT \leq \frac{1}{2}xOPT$. So, directly applying Lemma~\ref{lem_grtotaldist2parts}, we get that,

$$\frac{1}{n}w(B) = \frac{2}{N}w(B) \leq 2w(GR(B))(1 - 2x).$$

So, $\frac{1}{N}w(B) \leq xOPT(1-2x)$. Putting this inside the generic lower bound for $RD$, we complete the proof of this lemma. $\hfill\qed$
\end{proof}

We now have a bound for the case when $\alpha \leq \frac{1}{2}$. Next, we provide a universal bound for the other case $(\alpha \geq \frac{1}{2})$. Observe that in this case, $w(GR(B_1)) \geq w(GR(B_2))$. We leverage the low weight of $GR(B_2)$ to prove the following bound.

\begin{lemma}
\label{lem_rd_lowerbound2}
Suppose that for a given instance with $w(GR(B)) = xOPT$, $w(GR(B_1)) = \alpha xOPT$ with $\alpha \geq \frac{1}{2}$. Then, we have that

$$w(RD) \geq \frac{5}{8}OPT - x(1-\frac{3}{2}x)OPT.$$

\end{lemma}
\begin{proof}
Once again, we begin with a generic lower bound on $RD$ (Corollary~\ref{corr_random_lower3}) that depends on partitioning the node set $\mathcal{N}$ into $3$ parts ($T, B_1, B_2)$. Notice that $\frac{|B_2|}{N} = \frac{1}{2} - x$.

$$w(RD) \geq \frac{1}{2}OPT - \frac{x}{2}OPT + \frac{1}{N}\{w(T)  + \frac{1}{2}(w(T,B_1) - w(B_1,B_2)) - w(B_2)\}.$$

As with Lemma~\ref{lem_rd_lowerbound1}, we know that $\frac{w(T)}{N} \geq \frac{1}{8}OPT$. Now for every edge $e$ in $GR(T)$, note that the triangle inequality implies that for any node in $B_1$, going to that node from an endpoint of $e$ and coming back to the other endpoint of $e$ is larger than the weight of $e$. Summing these up, we get that $w(T,B_1)\geq |B_1|GR(T)$. Using the fact that $GR(T) \geq \frac{OPT}{2}$ gives a slightly simplified version.

$$w(RD) \geq \frac{5}{8}OPT - \frac{x}{4}OPT - \frac{1}{N} \{\frac{1}{2}w(B_1,B_2) + w(B_2)\}.$$

So now, it suffices to prove a lower bound on the negative quantities. From Lemma~\ref{lem_gr_crossupperbound}, we get that $w(B_1, B_2) \leq 2w(GR(B_1))|B_2| = 2\alpha xOPT(1/2 - x)N$.

Next, we have to provide an upper bound on $w(B_2)$ in order to complete the proof. We know as per our definitions of $GR(B_1)$, $GR(B_2)$ that each edge in the latter is no larger than the smallest edge in the former. Moreover, from Proposition~\ref{prop_simpleproprdgrper}, we know that $w^*= 2GR(B_1)/xN = \frac{2\alpha OPT}{N}$ is an upper bound on the weight of every edge inside $GR(B_2)$. So, we can directly turn to Lemma~\ref{lem_grupperequalweight} applied specifically to $GR(B_2)$ to obtain

$$w(B_2) \leq 2w(GR(B_2))(\frac{N}{2} - xN - t),$$

where $t = \frac{w(GR(B_2))}{w^*} = \frac{N(1-\alpha)xOPT}{2 \alpha OPT} = \frac{N(1-\alpha)x}{2\alpha}$. In conclusion, we have that

\begin{equation}\label{eq1}
\frac{1}{N}w(B_2) \leq (1-\alpha)xOPT(1 - 2x  - \frac{1-\alpha}{\alpha}x) \leq (1-\alpha)xOPT(1-2x).
\end{equation}

We are now ready to complete our (lower) bounds on the negative quantities

\begin{align*}
\frac{1}{2N}w(B_1,B_2) + \frac{1}{N}w(B_2) & \leq \alpha x OPT(1/2 - x) + (1-\alpha)xOPT(1-2x) \\
& = xOPT(1/2 - x)(\alpha + 2 - 2\alpha) \\
& \leq xOPT(1/2 - x)\frac{3}{2} & \text{(Since $\alpha \geq \frac{1}{2}$)}\\
& = \frac{3}{4}xOPT - \frac{3}{2}x^2OPT.
\end{align*}
Plugging the final inequality into the simplified generic lower bound completes the proof. $\hfill\qed$
\end{proof}

A careful inspection of the proof of Lemma~\ref{lem_rd_lowerbound2} reveals that our lower bound is a bit loose in two places where we independently replaced $\alpha$ with $\frac{1}{2}$ and $1$ respectively to provide a worst-case bound. Unfortunately, as a result, the lower bounds for the $\alpha \geq \frac{1}{2}$ and $\alpha \leq \frac{1}{2}$ cases do not align.

For our purposes however, it is enough to show that the two lower bounds apply when $x \leq \frac{1}{8}$, which we prove below in the third of our lemmas in this subsection.

\begin{lemma}
\label{lem_rd_lowerbound3}
Suppose that for a given instance with $w(GR(B)) = xOPT$ with $x \leq \frac{1}{8}$, we have $w(GR(B_1)) = \alpha xOPT$ with $\alpha \geq \frac{1}{2}$. Then, the following lower bound is true

$$w(RD) \geq \frac{5}{8}OPT - x(1-2x)OPT.$$

\end{lemma}
\begin{proof}
The proof of the lemma picks up from the previous Lemma~\ref{lem_rd_lowerbound2} with only a few simple tweaks. From Lemma~\ref{lem_rd_lowerbound3} (specifically using Inequality \ref{eq1}), we have that

$$\frac{1}{2N}w(B_1, B_2) + \frac{1}{N}w(B_2) \leq xOPT\{\alpha(1/2 - x) + (1-\alpha)(1 - x - \frac{x}{\alpha})\}.$$

From Lemma~\ref{lem_genericffunctiondiffernetial}, we know that the expression inside the curly parenthesis attains its maximum value for $\alpha = \frac{1}{2}$ in the given range of $x$. Therefore, substituting $\alpha = \frac{1}{2}$, we get

$$\frac{1}{2N}w(B_1, B_2) + \frac{1}{N}w(B_2) \leq OPT\{\frac{x}{4} - \frac{x^2}{2} + \frac{x}{2} - \frac{3}{2}x^2\}.$$

Directly plugging this upper bound into the simplified generic lower bound from Lemma~\ref{lem_rd_lowerbound2} is enough to prove the statement in the Lemma. $\hfill\qed$
\end{proof}

\textit{(Proof of Claim~\ref{clm_mainlbrandommatch})} The proof is a direct consequence of Lemmas~\ref{lem_rd_lowerbound1},~\ref{lem_rd_lowerbound2}, and~\ref{lem_rd_lowerbound3}. $\hfill\qed$

\subsection{Final Leg: Proving the Actual Bound}
Proposition~\ref{prop_grlowerboundfinal} and Claim~\ref{clm_mainlbrandommatch} are the only tools that we require to show the final bound. We prove this in two cases depending on whether or not $x \leq \frac{1}{8}$.

\subsubsection*{Case I: $x \leq \frac{1}{8}$}
Recall that we pick the random matching with probability $p = \frac{4}{7}$ and the greedy mathing with probability $1-p = \frac{3}{7}$. Suppose we use $w(M)$ to denote the weight of the matching returned by our algorithm. Then,

\begin{align*}
E[w(M)] & = (1-p)w(GR) + p\cdot w(RD) \\
& \geq OPT \{ (1-p)(\frac{1}{2} + x) + (p)(\frac{5}{8} - x + 2x^2)\} \\
& = OPT \{\frac{1}{2} + p\frac{1}{8} + x(1-2p) + 2px^2\} \\
\end{align*}

Since $p$ is fixed, it is not hard to see that the quantity $x(1-2p) + 2px^2$ is minimized at $x=\frac{1}{2} - \frac{1}{4p}$. Substituting $p=\frac{4}{7}$, we get $\frac{OPT}{E[w(M)]} \leq 1.7638$

\subsubsection*{Case I: $x \geq \frac{1}{8}$}
In this case, we need to use a weaker lower bound for $RD$.

\begin{align*}
E[w(M)] & = (1-p)w(GR) + pw(RD) \\
& \geq OPT \{ (1-p)(\frac{1}{2} + x) + (p)(\frac{5}{8} - x + \frac{3}{2}x^2)\} \\
& = OPT \{\frac{1}{2} + p\frac{1}{8} + x(1-2p) + \frac{3}{2}px^2\} \\
\end{align*}

Using basic calculus, we observe the expression in the final line is a non-decreasing function of $x$ in the range $[\frac{1}{8}, \frac{1}{2}]$ and so, its minimum value is attained at $x=\frac{1}{8}$. Substituting this value above, we get
$\frac{OPT}{E[w(M)]} \leq 1.7638.$
\end{proof}

\subsection{Max $k$-Matching}\label{sec:kmatching}
We now move on to the more general Max $k$-matching problem, where the objective is to compute a maximum weight matching consisting only of $k \leq \frac{N}{2}$ edges. Our previous results do not carry over to this problem. While we know from~\cite{anshelevichS16} that the greedy algorithm is half-optimal, one can easily construct examples where this is not truthful. On the other hand, the random matching algorithm is truthful but its approximation factor can be as large as $\frac{N}{k}$. Our main result in this section is based on the \textit{Random Serial Dictatorship} algorithm that in some sense combines the best of greedy and random into a single algorithm. Such algorithms have received attention for other matching problems~\cite{filosRatsikas16, filosRatsikasF014}; ours is the first result showing that these algorithms can approximate the optimum matching up to a small constant factor for metric settings. Specifically, while serial dictatorship is usually easy to analyze, our algorithm greatly exploits the randomness to select good edges \emph{in expectation}.

\noindent{\bf Definition}: \emph{Random Serial Dictatorship} is the same algorithm as Serial Dictatorship (Algorithm \ref{alg_sd}), except the agents $x$ from $T$ are picked uniformly at random.


\begin{theorem}
\label{thm_rsdmaxk}
Random serial dictatorship is a universally truthful mechanism that provides a $2$-approximation for the Max $k$-matching problem.
\end{theorem}

\section{Truthful Mechanisms for Other Problems}\label{sec:truthfulother}

\subsection{Densest $k$-Subgraph}

In this section we present our truthful, ordinal algorithm for Densest $k$-subgraph, which requires techniques somewhat different from the ones outlined in Section \ref{sec:prelim}. While ``conventional" approaches such as Greedy and Serial Dictatorship {\em do} lead to good approximations for this problem, they are not truthful, whereas random approaches are truthful but result in poor worst-case approximation factors. We combat this problem with a somewhat novel approach that combines the best of both worlds by designing a semi-oblivious algorithm that has the following property: if agent $i$ is included in the solution, then changing her preference ordering $s_i$ does not affect the mechanism's output.

\begin{algorithm}[htbp]
{$S:= \emptyset$, $T$ is the set of available agents initialized to $\mathcal{N}$}\;
\While{$|S| < k$}{
pick an anchor agent $a$ and another node $x$, both uniformly at random from $T$\;
let $b$ denote $a$'s most preferred agent in $T - \{a,x\}$\;
with probability $\frac{1}{2}$, add $a,x$ to $S$, and set $T = T -\{a,x\}$\;
with probability $\frac{1}{2}$, add $b,x$ to $S$ and set $T = T -\{a,b,x\}$\;
}
\label{alg_dks}
\caption{Hybrid Algorithm for Densest $k$-Subgraph}
\end{algorithm}

\begin{theorem}
Algorithm~\ref{alg_dks} is a universally truthful mechanism that yields a $6$-approximation for the Densest $k$-Subgraph problem.
\end{theorem}

\noindent To see why this is truthful, note that for any particular choice of the anchor agent $a$, the only case in which $a$'s preference ordering makes a difference is when $a$ is definitely not added to the final team. Therefore, by lying $a$ cannot influence her utility in the event that she is actually chosen. 

\textbf{Remark on size of $k$} Without loss of generality, we assume that $k \leq \frac{N}{2}$ so that $T$ does not become empty before $|S| = k$. When $k \geq \frac{N}{2}$, there is a trivial algorithm that yields a $6$-approximation to the optimum densest subgraph (see Appendix~\ref{app:dks}). Since we are interested in asymptotic performance bounds, we also assume that $k$ is even. For the rest of this proof, given any set $S$, node $x$, $w(S)$ will denote the total weight of the edges inside $S$, and $w(x,S) := \sum_{j \in S}w(x,j)$.

\begin{proof}
\noindent\textbf{Notation}
We begin by defining some notation pertinent to the analysis. Suppose that our algorithm proceeds in rounds such that in each round, exactly two nodes are added to our set $S$, and at most $3$ nodes are removed from $\mathcal{A}$. Therefore, $S$ consists of $2j$ nodes after $j$ rounds. For ease of notation, we will number the rounds $2,4,6,\ldots$ instead of $1,2,3,4,5,\ldots$; thus $S$ has $r$ nodes at the end of round $r$. Further, define $S_r$ to be the random set of selected nodes after round $r$, i.e., $|S_r| = r$ for any instantiation of this random set.

Next, let us examine the inner workings of the algorithm. Look at any round $r$, the algorithm works by selecting a triplet $\Delta_r = \{a,x,b\}$, where $a$ is referred to as the anchor node, $x$ is a node selected uniformly at random, and $b$ is $a$'s most preferred agent in $\mathcal{A}_r - \{a,x\}$, (let $\mathcal{A}_r$ be the random set of available nodes at the beginning of round $r$). For the rest of this proof, we will use $\Delta_r$ to denote the random triplet of nodes selected in round $r$. Notice that for a given (ordered) triplet $\{a,x,b\}$, the algorithm adds $(a,x)$ to $S$ with probability half and $(b,x)$ to $S$ also with the same probability.

Finally, we use $OPT_r$ to denote the weight of the optimum solution to the Densest $k$-subgraph problem when $k=r$, and $Alg_r$ to be the expected weight of the solution output by our algorithm for the same cardinality, i.e., $Alg_r=E[w(S_r)]$. Let $alg_{r+2}$ represent the expected increase in the weight of the solution output by our algorithm from $r$ to $r+2$, i.e., $alg_{r+2} = Alg_{r+2} - Alg_r$. We will prove by induction on even $r$ that $OPT_r \leq 6Alg_r$. More specifically, we will show that for each $r$, $OPT_r - OPT_{r-2} \leq 6alg_r$.

\subsubsection*{Proof by Induction: $OPT _r \leq 6Alg_r$}

\begin{claim}(\textbf{Base Case:} $r=2$)
$OPT_2 \leq 4 Alg_2$.
\end{claim}
\begin{proof}
The base case is quite straightforward.
Suppose that $w^*_{max}$ is the heaviest edge in $\mathcal{N}$. Clearly, $OPT_2 = w^*_{max}$. Next, let $a,x \in \mathcal{N}$ be any two agents, and let $b$ denote $a$'s most preferred agent in $\mathcal{N} - \{x\}$. Then, we claim that $w(a,x) + w(b,x) \geq \frac{1}{2}w^*_{max}$.

The above claim can be proved in two cases: first, suppose that $b$ is indeed $a$'s favorite node in $\mathcal{N}$. Then, as per Lemma~\ref{lem_rsdedgebound}, $w(a,x) + w(b,x) \geq w(a,b) \geq \frac{1}{2}w^*_{max}$. In the second case, if $a$'s most preferred node in $N$ and $N-\{x\}$ do not coincide, the only possibility is that $x$ is $a$'s most preferred node in $\mathcal{N}$, and by the same lemma $w(a,x) \geq \frac{1}{2}w^*_{max}$.

To complete the base case, consider any instantiation of the random triplet, $\Delta_2 = \{a,b,x\}$. We have that $S_2 = \{a,x\}$ with probability $\frac{1}{2}$ and $S_2 =\{b,x\}$ otherwise. Therefore, for this instantiation $w(S_2) = \frac{1}{2}(w(a,x) + w(b,x)) \geq \frac{1}{4}w^*_{max}$. Taking the expectation over every such triplet, we get the desired base claim. \hfill$\qed$

\end{proof}

\subsection*{Inductive Claim: To Prove $OPT_{r+2}\leq 6Alg_{r+2}$}
Recall that $S_r$ denotes the random set of chosen nodes at the end of round $r$. We know from the induction hypothesis that $OPT_{r} \leq 6Alg_r = 6E[w(S_r)]$. Consider some specific instantiation of $S_r$, call it $\bar{S}_r$, and for this instantiation, let $\bar{\Delta}_{r+2} = \{a,x,b\}$ denote some random triplet selected by the algorithm in round $r+2$, i.e., we have a specific instantiation of $S_r$ and $\Delta_{r+2}$ for our algorithm. As usual, for this triplet, $(a,b,x)$, $a$ is the anchor node, $x$ is the random node and $b$ is $a$'s most preferred node in $\mathcal{N} \setminus \{\bar{S}_r \cup \{a,x\}\}$.

Suppose that $\overline{alg}_{r+2}$ is the increase in the expected weight of the solution returned by our algorithm during round $r+2$ for this specific instantiation of $\bar{S}_r, \bar{\Delta}_{r+2}$., i.e., $\overline{alg}_{r+2} = \frac{1}{2}[w(\bar{S}_r \cup \{a,x\}) + w(\bar{S}_r \cup \{b,x\})] - w(\bar{S}_r)$. Our proof will proceed as follows: we establish an upper bound for $OPT_{r+2} - OPT_r$ in terms of $\overline{alg}_{r+2}$, and then take the expectation over all possible instantiations to get the actual bound.

Before starting with the proof of the inductive claim, we define some auxiliary notation that will allow us to process $OPT$ as a sequence of additions in each round, so that we can compare the addition to $OPT$ in round $r+2$ to that of our algorithm in the same round. Fix $p,q$ to be any two nodes in $OPT_{r+2} \setminus \bar{S}_r$, and let $T := OPT_{r+2} \setminus \{p,q\}$. $T$ will act as a proxy to $OPT_r$ in our proofs. Notice that $p,q \in \mathcal{N} \setminus \bar{S}_r$. Finally, in order to avoid messy notation, assume that $b$ is $a$'s (most) preferred node in $\mathcal{N} \setminus \bar{S}_r$. If this is not the case (and this can happen with a small probability), then $a$'s most preferred node in $\mathcal{N} \setminus \bar{S}_r$ \emph{has to be} $x$. We deal with this case separately in Section~\ref{subsec_inductivebadcase} although the proof is quite similar.

We begin with a nice lower bound for $\overline{alg}_{r+2}$. Suppose that $w(a,b) = w^*_a$.
\begin{lemma}\label{lem:alglb}
(Lower Bound for our Algorithm)
$$\overline{alg}_{r+2} \geq \frac{1}{6}[w(a,\bar{S}_r \cap T) + w(b,\bar{S}_r \cap T)] + \frac{1}{3}|\bar{S}_r \cap T|w^*_a + \frac{1}{2}(|\bar{S}_r \setminus T|+1)w^*_a + \frac{1}{r-1}w(\bar{S}_r).$$
\end{lemma}
\begin{proof}
Recall that $\overline{alg}_{r+2} = \frac{1}{2}[w(\bar{S}_r \cup \{a,x\}) + w(\bar{S}_r \cup \{b,x\})] - w(\bar{S}_r)$. Simplifying the expression, we get

\begin{equation}
\label{eqn_dks_genlowerb}
\overline{alg}_{r+2} = \frac{1}{2}[w(a,\bar{S}_r) + w(b,\bar{S}_r) + w(a,x) + w(b,x)] + w(x,\bar{S}_r)
\end{equation}

Consider the first two terms inside the square brackets. We can divide $\bar{S}_r$ into $\bar{S}_r \cap T$ and $\bar{S}_r \setminus T$ and simplify the two parts as follows,

$$w(a,\bar{S}_r \cap T) + w(b,\bar{S}_r \cap T) \geq \frac{1}{3}[w(a,\bar{S}_r \cap T) + w(b,\bar{S}_r \cap T)] + \frac{2}{3}|\bar{S}_r \cap T|w^*_a.$$

The right most term in the RHS simply comes from the triangle inequality since for any $i \in \bar{S}_r \cap T$, $w(i,a) + w(i,b) \geq w(a,b) = w^*_a$. Now for the second part, which also follows from the triangle inequality,
$$w(a,\bar{S}_r \setminus T) + w(b,\bar{S}_r \setminus T) \geq |\bar{S}_r \setminus T|w^*_a.$$

To wrap up the proof, we apply Lemma~\ref{lem_pointtoset} to $w(x,\bar{S}_r)$ to get $w(x,\bar{S}_r) \geq \frac{1}{r-1}w(\bar{S}_r)$. An additional $\frac{1}{2}w^*_a$ can extracted from $\frac{1}{2}[ w(a,x) + w(b,x)]$. Adding up the various parts completes the lemma. \hfill$\qed$
\end{proof}

Before showing our upper bound on $OPT_{r+2} - OPT_r$, we present a simple lemma that allows us to relate the weights of any given node to the members of a set in terms of $w^*_a$ and the weight of $a$ to the members of that set. Recall the definitions of $p,q \in OPT_{r+2} \setminus \bar{S}_r$.

\begin{lemma}
\label{lem_dkssetcharging}
Suppose that $T,p,q$ are as defined previously. Then,

\begin{enumerate}
\item $w(p,T) \leq w(a,T \cap \bar{S}_r) + |T \cap \bar{S}_r|w^*_a +  2|T \setminus \bar{S}_r|w^*_a$.

\item $w(q,T) \leq w(b,T\cap \bar{S}_r) + 2|T \cap \bar{S}_r|w^*_a + 2|T \setminus \bar{S}_r|w^*_a$.
\end{enumerate}
\end{lemma}
\begin{proof}
\textbf{(Part I)} The proof proceeds as follows: remember that since $b$ is $a$'s most preferred node in $\mathcal{N} \setminus \bar{S}_r$, for any $i \notin \bar{S}_r$, $w(i,a) \leq w^*_a$. This includes $i = p$. Moreover, for any $i,j \notin \bar{S}_r$, $w(i,j) \leq 2w^*_a$ as per Lemma~\ref{lem_rsdedgebound}.
\begin{align*}
w(p,T) & = \sum_{j \in T \cap \bar{S}_r}w(p,j) + \sum_{j \in T \setminus \bar{S}_r}w(p,j)\\
& \leq \sum_{j \in T \cap \bar{S}_r}[w(p,a) + w(a,j)] + \sum_{j \in T\setminus \bar{S}_r}2w^*_a\\
& \leq \sum_{j \in T \cap \bar{S}_r}[w^*_a + w(a,j)] + 2|T \setminus \bar{S}_r|w^*_a\\
& \leq \sum_{j \in T \cap \bar{S}_r}w(a,j) + |T \cap \bar{S}_r|w^*_a + 2|T \setminus \bar{S}_r|w^*_a.
\end{align*}

\textbf{(Part II)} The proof of the second part is almost the same as the first, except that for any $i \notin \bar{S}_r$, we have that $w(b,i) \leq 2w^*_a$, once again as the product of Lemma~\ref{lem_rsdedgebound}.
\begin{align*}
w(q,T) & = \sum_{j \in T \cap \bar{S}_r}w(q,j) + \sum_{j \in T \setminus \bar{S}_r}w(q,j)\\
& \leq \sum_{j \in T \cap \bar{S}_r}[w(q,b) + w(b,j)] + \sum_{j \in T\setminus \bar{S}_r}2w^*_a\\
& \leq \sum_{j \in T \cap \bar{S}_r}[2w^*_a + w(b,j)] + 2|T \setminus \bar{S}_r|w^*_a\\
& \leq \sum_{j \in T \cap \bar{S}_r}w(b,j) + 2|T|w^*_a. \hfill\qed
\end{align*}

\end{proof}

\subsubsection*{Upper Bound on $OPT_{r+2}$ to Complete the Inductive Claim}

Now we express $OPT_{r+2} - OPT_r$ in terms of $\overline{alg}_{r+2}$ for the given instantiation.
\begin{lemma}
$$OPT_{r+2} - OPT_r \leq 6\overline{alg}_{r+2} + \frac{1}{r-1}OPT_r - \frac{6}{r-1}w(\bar{S}_r).$$
\end{lemma}
\begin{proof}
Consider $OPT_{r+2}$ and remember that by definition $p,q \notin \bar{S}_r$. We can divide up $OPT_{r+2}$ in two ways.

\begin{align*}
OPT_{r+2} & = w(T) + w(p,q) + w(p,T) + w(q,T)\\
&  \leq  OPT_{r} + w(p,q) + w(p,T) + w(q,T). & \text{(First)}\\\\
OPT_{r+2} & = w(T \cup \{q\}) + w(p,q) + w(p,T) \\
& \leq  OPT_{r+1} + w(p,q) + w(p,T). & \text{(Second)}
\end{align*}

Clearly, $w(T) \leq OPT_r$ and $w(T \cup \{q\}) \leq OPT_{r+1}$. Further, applying Lemma~\ref{lem_dkskincreases}, we get that $OPT_{r+1} \leq OPT_r + \frac{2}{r-1}OPT_r.$ Using this to simplify the second inequality, we then add the simplified inequalities above and divide by two to get:

\begin{equation}
\label{eqn_optupper}
OPT_{r+2} \leq OPT_r + \frac{1}{r-1}OPT_r + w(p,q) + w(p,T) + \frac{1}{2}w(q,T).
\end{equation}

Next, we simplify the (two) rightmost terms in the RHS to reflect their dependence on $w^*_a$, which we recall is the weight of the maximum weight edge containing $a$ in $\mathcal{N} \setminus \bar{S}_r$. Applying Lemma~\ref{lem_dkssetcharging} (Part 1), we get

\begin{equation}
\label{eqn_p_upper}
w(p,q) + w(p,T) \leq 2w^*_a + w(a,\bar{S}_r \cap T) + |T \cap \bar{S}_r|w^*_a + 2|T \setminus \bar{S}_r|w^*_a.
\end{equation}
Similarly, using the second half of Lemma~\ref{lem_dkssetcharging}, we bound $w(q,T)$.
\begin{equation}
\label{eqn_q_upper}
\frac{1}{2}w(q,T) \leq \frac{1}{2}w(b, T \cap \bar{S}_r) + |T \cap \bar{S}_r|w^*_a + |T \setminus \bar{S}_r|w^*_a.
\end{equation}

Adding Equations~\ref{eqn_p_upper} and~\ref{eqn_q_upper}, and substituting the result into Equation~\ref{eqn_optupper}, we have that

$$OPT_{r+2} - OPT_{r} \leq \frac{1}{r-1}OPT_r + 3(|T \setminus \bar{S}_r| + 1)w^*_a + w(a,\bar{S}_r \cap T) + w(b,\bar{S}_r \cap T) + 2|T \cap \bar{S}_r|w^*_a$$

Recall from Lemma \ref{lem:alglb} that $\overline{alg}_{r+2} \geq \frac{1}{6}w(a, T \cap \bar{S}_r) + w(b, T \cap \bar{S}_r)] + \frac{1}{3}|T \cap \bar{S}_r|w^*_a + \frac{1}{2}(|T \setminus \bar{S}_r| + 1)w^*_a + \frac{1}{r-1}w(\bar{S}_r)$. Therefore, we get our required lemma by writing the RHS of our upper bound on $OPT_{r+2} - OPT_{r}$ in terms of $6\overline{alg}_{r+2}$.

$$OPT_{r+2} - OPT_r \leq \frac{1}{r-1}OPT_r+ 6\overline{alg}_{r+2} - \frac{6}{r-1}w(\bar{S}_r) \hfill\qed$$

\end{proof}

Having wrapped up our upper bound, we are ready to prove our actual inductive claim.

From our upper bound, we have that for every possible realization of $S_r$, $\Delta_{r+2}$, we know that $OPT_{r+2} - OPT_r \leq \frac{1}{r-1}OPT_r+ 6\overline{alg}_{r+2} - \frac{6}{r-1}w(\bar{S}_r).$ Now, we push to complete our proof,

$$OPT_{r+2} - OPT_r \leq \frac{1}{r-1}OPT_r + E_{S_r, \Delta_{r+2}}[6\overline{alg}_{r+2} - \frac{6}{r-1}w(S_r)].$$

The term inside the expectation is clearly $6alg_{r+2} - \frac{6}{r-1}E[w(S_r)] \leq 6alg_{r+2} - \frac{1}{r-1}OPT_{r}$; the final inequality is a result of the induction hypothesis. The term $\frac{1}{r-1}OPT_{r}$ cancels out, giving us our desired claim

$$OPT_{r+2} \leq OPT_r + 6alg_{r+2} \leq 6E[w(S_r)] + 6alg_{r+2} = 6E[w(S_{r+2})] = 6Alg_{r+2}. \hfill\qed $$

%

Therefore our hybrid algorithm for Densest $k$-subgraph always returns a solution that is within a sixth of the optimum densest subgraph.

\subsection{Inductive Claim when $w^*_a = w(a,x) > w(a,b)$}
\label{subsec_inductivebadcase}

Suppose that $x$ is $a$'s most preferred node in $\mathcal{N} \setminus \bar{S_r}$. Claim for claim, the proof proceeds in the same way as above except that the role played by $b$ in the previous proof is now played by $x$. We go over the proof of this case for completeness. Assume the same notation as before, and consider the following lower bound on $\overline{alg}_{r+2}$.

\subsubsection*{Part 1: Lower Bound}

\begin{equation}
\label{eqn_altcase1}
\overline{alg}_{r+2} \geq \frac{1}{6}[w(a, T \cap \bar{S}_r) + w(x,T \cap \bar{S}_r)] + \frac{1}{3}|T \cap \bar{S}_r|w^*_a + \frac{1}{2}(|T \setminus \bar{S}_r | + 1)w^*_a + \frac{1}{r-1}w(\bar{S}_r)
\end{equation}

To see why Equation~\ref{eqn_altcase1} is true: first notice that $w(a,T \cap \bar{S}_r) + w(x,T\cap \bar{S}_r) \geq |T \cap \bar{S}_r|w^*_a$, and $w(a,T \setminus \bar{S}_r) + w(x,T\setminus \bar{S}_r) \geq |T \setminus \bar{S}_r|w^*_a$ by a direct application of the triangle inequality. Therefore, from these two inequalities we glean that

$$\frac{1}{2}[w(a,\bar{S}_r) + w(x,\bar{S}_r)] \geq \frac{1}{6}[w(a,T \cap \bar{S}_r) + w(x,T \cap \bar{S}_r)] + \frac{1}{3}|T \cap \bar{S}_r|w^*_a + \frac{1}{2}|T \setminus \bar{S}_r|w^*_a.$$

The remaining terms in $\overline{alg}_{r+2}$ are $\frac{1}{2}[w(x,\bar{S}_r) + w(b,\bar{S}_r) + w(a,x) + w(b,x)]$: $(i)$ From Lemma~\ref{lem_pointtoset}, $\frac{1}{2}w(x,\bar{S}_r) \geq \frac{1}{2(r-1)}w(\bar{S}_r)$, $(ii)$ Via the same lemma, $\frac{1}{2}w(b,\bar{S}_r) \geq \frac{1}{2(r-1)}w(\bar{S}_r)$, and $(iii)$ $\frac{1}{2}(w(a,x) + w(b,x)) \geq \frac{1}{2}w^*_a$ since $w(a,x) = w^*_a$. Adding $(i)$, $(ii)$, $(iii)$ with our lower bound for $\frac{1}{2}[w(a,\bar{S}_r) + w(x,\bar{S}_r)]$ completes the proof of the first part, i.e., Equation~\ref{eqn_altcase1}.

\subsubsection*{Part 2: Simplifying Lemma}

Now, we make the second claim for set $T$ as defined previously and $q \notin \bar{S}_r$, also as defined previously.

\begin{equation}
\label{eqn_altcase2}
w(q,T) \leq w(x, T \cap \bar{S}_r) + 2|T \cap \bar{S}_r|w^*_a + 2|T \setminus \bar{S}_r|w^*_a.
\end{equation}

The proof is exactly the same as in Lemma~\ref{lem_dkssetcharging} so we do not restate it. Once again, the main observation here is that for any $j \notin \bar{S}_r$, $w(q,j) \leq 2w^*_a$ from Lemma~\ref{lem_rsdedgebound}.

\subsubsection*{Part 3: Upper Bound}

Now we are ready to complete the proof. Let us begin by restating Equation~\ref{eqn_optupper}, which is a generic condition and does not depend on $a$, $x$ or $b$: $$OPT_{r+2} \leq OPT_r + \frac{1}{r-1}OPT_r + w(p,q) + w(p,T) + \frac{1}{2}w(q,T).$$

From Lemma~\ref{lem_dkssetcharging}, we have that $(iv)$ $w(p,q) + w(p,T) \leq w(a, T \cap \bar{S}_r) + |T \cap \bar{S}_r|w^*_a + 2|T \setminus \bar{S}_r + 1|w^*_a$; $(v)$ From Equation~\ref{eqn_altcase2}, $\frac{1}{2}w(q,T) \leq w(x, T \cap \bar{S}_r)+ |T \cap \bar{S}_r|w^*_a + |T \setminus \bar{S}_r|w^*_a.$

Adding these two equations, the rest of the proof follows as from before. 
\end{proof}

\subsection{A $2$-approximation algorithm for $k$-Sum Clustering}\label{sec:ksum}
In the literature, the $k$-sum clustering problem has only been studied in a full information setting, sometimes amidst the class of dispersion problems~\cite{hassin1997approximation}. The best known approximation algorithm for this is a $2$-approximation that uses the optimum matching as an intermediate. Instead, we give a \emph{much simpler} algorithm with the same factor that is completely oblivious to the input, and is therefore truthful. Although the analysis of the algorithm involves new upper bounds on the optimum solution, it is still not difficult, so we include this result mostly for completeness.

Recall that in the Max $k$-Sum problem, we are provided as input a vector of preference lists $P(\mathcal{N})$ along with a positive integer $2 \leq k \leq \frac{N}{2}$ with the objective being to partition the set of points $\mathcal{N}$ into $k$ equal-sized clusters (of size $\gamma = \frac{N}{k}$; we assume that $N$ is divisible by $k$) $S=(S_1, \ldots, S_k)$ to maximize $\sum_{i=1}^k \sum_{x,y \in S_i} w(x,y).$

\begin{theorem}
\label{thm:maxksum}
There exists an ordinal universally truthful $2$-approximation algorithm for the $k$-sum clustering problem.
\end{theorem}
\begin{proof}
We use a simple approach that picks sets (clusters) of size $\gamma$ uniformly at random.
\begin{enumerate}
\item For $i=1$ to $k$
\item Choose $\gamma$ nodes uniformly at random from $\mathcal{N}$.
\item Remove these nodes from $\mathcal{N}$, and add them to $S_i$.
\item Return the final solution $S$.
\end{enumerate}


\begin{lemma}
(Lower Bound) The expected value of the objective function for the clustering returned by our algorithm $(S_1, \ldots, S_k)$ is exactly
$$\frac{\gamma-1}{N-1}\sum_{(x,y) \in \mathcal{N} \times \mathcal{N}}w(x,y).$$
\end{lemma}
\begin{proof}
We proceed via a symmetry argument although it is not hard to verify that the same bound can be obtained via a more exhaustive counting argument. First, by linearity of expectation, we have that the value of the objective (in expectation) is $\sum_{(x,y) \in \mathcal{N} \times \mathcal{N}}w(x,y)Pr(x,y \in S_i)$ where the second term is the probability that $x$ and $y$ belong to the same cluster in $S$. Using a symmetry argument (since our process chooses edges uniformly at random), we claim that the probability $Pr(x,y \in S_i)$ is the same for every $x,y \in \mathcal{N}$.

Now, fix any arbitrary node $x \in \mathcal{N}$: since there $\gamma-1$ other nodes in the same cluster as $x$, this means that $\sum_{y \neq x}Pr(x,y \in S_i) = \gamma-1$. Therefore, for every $(x,y)$, $Pr(x,y \in S_i) = \frac{\gamma-1}{N-1}$. Substituting this in the expected value of the objective function gives us the desired result.
\end{proof}

\begin{lemma}
(Upper Bound) Suppose that $O=(O_1, \ldots, O_k)$ is the optimum solution for a given instance of the Max $k$-sum problem. Then, we have the following upper bound on the value of the optimum solution
$$\sum_{i=1}^k \sum_{x,y \in O_i}w(x,y) \leq \frac{2(\gamma-1)}{N-1}\sum_{(x,y) \in \mathcal{N} \times \mathcal{N}}w(x,y).$$
\end{lemma}
\begin{proof}
Suppose that $x$ and $y$ are two nodes belonging to the same cluster in $O$. Then, by the triangle inequality, we have that for every $z \in \mathcal{N}$ (including $x$ and $y$), $w(x,z) + w(y,z) \geq w(x,y)$. Summing this up over all $z \in \mathcal{N}$, we have $\sum_{z \in \mathcal{N}}(w(x,z) + w(y,z)) \geq Nw(x,y)$. Repeating this process over all $(x,y) \in S$ and $z \in \mathcal{N}$, we get

\begin{align*}
\sum_{i=1}^k \sum_{x,y \in S_i} \sum_{z \in \mathcal{N}}(w(x,z) + w(y,z)) & \geq N \sum_{i=1}^k \sum_{x,y \in O_i}w(x,y)\\
& = N OPT.
\end{align*}

Now, given some edge $w(x,z)$, how many times does this edge appear in the LHS? Without loss of generality, suppose that $x \in O_i$ and $z \in O_j$. Then, $x$ has $\gamma-1$ edges inside $O_i$ and $w(x,z)$ appears once in the LHS for each of these neighboring edges. Similarly, $z$ has $\gamma-1$ edges inside $O_j$ and $w(x,z)$ appears once in the LHS for each edge. Therefore, for every $x,z \in \mathcal{N}$, $w(x,z)$ appears $2(\gamma-1)$ times in the LHS of the above equation. Substituting this, we prove the lemma,
$$\sum_{x,y \in \mathcal{N}}2(\gamma-1)w(x,y) \geq NOPT.$$

\end{proof}
The rest of the theorem follows immediately from the two lemmas.
\end{proof}

\subsection{Max Traveling Salesman Problem}\label{sec:truthfulTSP}
The max traveling salesman problem has received a lot of attention in the literature despite not being as popular as the minimization variant, and has seen a plethora of algorithms for both the metric and the non-metric versions \cite{kosarajuPS94,kowalik2009deterministic}. Such algorithms usually work by looking at the optimum matching and cycle cover and cleverly interspersing the two solutions to form a Hamiltonian cycle. In adapting this approach to our setting, we would be bottlenecked by the best possible ordinal algorithms for the above two problems. Instead, we take a direct approach towards computing a tour and show that a simple algorithm based on Serial Dictatorship results in a $2$-approximation factor.

\begin{algorithm}[htbp]
{Initialize $T$ to be a random edge from the complete graph on $\mathcal{N}$\;
 Let $S$ be the set of available agents initialized to $\mathcal{N}$}\;
\While{$S \neq \emptyset$}{
pick one of the end-points of $T$, say $x$ \;
let $y$ denote $x$'s most preferred agent in $S$;
add $(x,y)$ to $T$ and remove $y$ from $S$\;
}
Complete $T$ to form a Hamiltonian cycle\;
\label{alg_tsp}
\caption{Serial Dictatorship for Max TSP}
\end{algorithm}

\begin{theorem}
Algorithm~\ref{alg_tsp} is a universally truthful mechanism that provides a $2$-approximation to the optimum tour. Moreover, the algorithm provides a $(2+\epsilon)$-approximation, where $\epsilon \to 0$ as $N \to \infty$, even when the edge weights do not obey the metric assumption.
\end{theorem}

\noindent It is easy to see that this algorithm is truthful: when an agent $i$ is asked for its preferences, the first edge of $T$ incident to agent $i$ has already been decided, so $i$ cannot affect it. Thus, to form the second edge of $T$ incident to $i$, it may as well specify its most-preferred edge. Note that the randomization in the first step is {\em essential:} if we had selected the first edge based on the input preferences, then the first node could improve its utility by lying, and the algorithm would no longer be strategy-proof.

\section{Non-Truthful Ordinal Mechanisms}
In this section we consider the case when agents are not able to lie, i.e., the algorithm knows their true ordinal preferences $P_i$, but is still ignorant of the hidden underlying metric utilities which induce those preferences. Designing algorithms for this setting captures the true power of ordinal information, as the necessity for approximation arises from the fact that the algorithm only has limited ordinal information at its disposal, as opposed to the agents being self-interested. This can occur due to the fact that specifying ordinal preferences is much easier than specifying the full numerical utility information; in fact even in the case when such latent numerical utilities exist, it may be difficult for the agents to quantify them precisely. This is the setting studied in papers such as \cite{anshelevichBP15,boutilierCHLPS15,procacciaR06,randomized,regret}; in \cite{anshelevichS16}, the authors previously gave ordinal approximation algorithms for Densest $k$-Subgraph and Max TSP with approximation factors of 4 and 2.14; here we improve the TSP approximation factor to 1.88 and give a new ordinal bicriteria approximation algorithm which shows that by relaxing the set size $k$ by a small amount for Densest $k$-Subgraph, a much better approximation can be achieved.

\subsection{Densest $k$-Subgraph}
\subsubsection*{$4$-approximation Algorithm}
We begin by presenting an extremely simple $4$-approximation algorithm for this problem; while not explicitly mentioned there, it was alluded to in \cite{anshelevichS16}. Given an input $k$, the algorithm essentially computes a $2$-approximate maximum matching with exactly $\frac{k}{2}$ edges using the algorithm from Theorem~\ref{thm_rsdmaxk}. We then show that the $k$ nodes that make up these edges provide a $4$-approximation to the optimum solution for this problem. Unfortunately, despite the truthfulness of RSD for Max $k$-matching, it is easy to construct examples where this mechanism is no longer truthful for Densest subgraph.

\subsubsection*{Bicriteria Approximation}
\label{sec_bicriteria}
Among the problems that we study, Densest $k$-subgraph naturally lends itself to bicriteria approximation algorithms. For instance, when we construct committees, a little additional leverage on the committee size may lead to much more diverse committees. Formally, given a parameter $k$, an algorithm for Densest $k$-subgraph is said to be a bi-criteria $(\alpha, \beta$) approximation if the objective value of the solution $S$ output by the algorithm for every instance is at least a factor $\frac{1}{\alpha}$ times that of the optimum solution of size $k$, and if $|S| \leq \beta k$, for $\alpha, \beta \geq 1$.
Here, we give bounds on how $\alpha$ decreases when $\beta$ increases. In particular, we show that if we are allowed to choose a committee of size $2k$, the value of our solution is equal to the optimum solution of size $k$. If instead we {\em must} form a committee of size exactly $k$, then this results in an ordinal 4-approximation algorithm.

\begin{theorem}
We can efficiently compute an ordinal $(\frac{4}{\beta^2}, \beta)$-approximate solution for the Densest $k$-subgraph problem for $\beta \leq 2$, i.e., a solution of size $\beta k$, whose value is at least $\frac{\beta^2}{4}$ times that of the optimum solution of size $k$.
\end{theorem}

\noindent The algorithm that provides us the approximation factor is simple: we compute a greedy matching of size $\frac{\beta k}{2}$, and return its endpoints. However the analysis is quite involved. One of the salient features of this result is that a small change in $\beta$ results in a super-linear improvement in efficiency.
For example, in order to obtain a $2$-approximation to the Densest $k$-subgraph, it is enough to compute a set of size $\sim 1.4k$.

\vskip 2pt\noindent{\em Proof Sketch:~}  The proof involves carefully charging different sets of node distances in the optimal solution to node distances in our solution. So, before giving the main proof, we provide a series of very general charging lemmas. We define a new helpful tool which we call the {\em top-intersecting matching}; we are able to use this to establish various bounds which yield our result for Densest $k$-subgraph. We believe that this tool and the bounds we show using it may be useful in forming other ordinal approximations.

Specifically, given a matching $M$, we will use $N(M)$ to denote the set of nodes which form the endpoints of the edges in $M$. Suppose that we are provided a matching $M$ of some given size, and a set $B \subseteq N(M)$. Now, given an integer $t \leq |B|$, define $M(t,B)$ to be the top (i.e., highest weight) $t$ edges in $M$, such each edge in $M(t,B)$ contains at least one node from $B$. We refer to $M(t,B)$ as the top-intersecting matching. In the proof, we highlight the versatility of the top-intersecting matching by charging different sets of inter-node distances to this matching. Afterwards, we use these charging lemmas to prove the main theorem. To give the flavor of these arguments, we provide some of the upper bounds below. Here we assume that $M$ is a greedy matching of size $k$, initialized with the complete edge set.

\begin{lemma}
Suppose that $M$ is a greedy matching, and suppose that $B$ and $C$ are two disjoint sets such that $B \subseteq N(M)$ with $|B|=2m$, and $C \cap N(M) = \emptyset$. Then the following bounds hold,
\begin{eqnarray*}
\sum_{x,y \in B}w(x,y) \leq \sum_{x,y \in N(M(m,B)) \cap B}w(x,y) + \frac{5r}{2}w(M(m,B))\\
\sum_{x \in B, y \in C}w(x,y) \leq 2|C|w(M(m,B))\\
\sum_{x ,y \in C}w(x,y) \leq \frac{(|C|^2)}{|M| - m}w(M \setminus M(m,B))
\end{eqnarray*}
where $r = |B \setminus N(M(m,B))|$.
\end{lemma}

\subsection{Max Traveling Salesman}
\label{sec:TSP}
We now present an ordinal (but not truthful) $1.88$-approximation algorithm for Max TSP. Unlike most of the algorithms in this paper, this algorithm is rather complex, since it requires carefully balancing several different tour constructions.

\begin{theorem}
\label{thm:TSPimproved}
We give an ordinal and efficient approximation algorithms for Max TSP whose approximation factor approaches $\frac{32}{17}\approx 1.88$ as $N \to \infty$.
\end{theorem}

\noindent{\em Proof Sketch:~}  Before defining our randomized algorithm, we first present the following lemmas: one gives a relationship between matching and Hamiltonian paths and the other shows how to stitch together two paths to form a good tour using only ordinal information.

\begin{lemma}
\label{lem_tourcompletion_body}
Given any matching $M$ with $k$ edges, there exists an efficient ordinal algorithm that computes a Hamiltonian path $Q$ containing $M$ such that the weight of the Hamiltonian path in expectation is at least
$$[\frac{3}{2} - \frac{1}{k}]w(M).$$
\end{lemma}

\begin{lemma}
\label{lem_pathtotour_body}
Let $H_1$ and $H_2$ be two Hamiltonian paths on two different sets of nodes, with $a$ and $b$ the endpoints of $H_1$. Then, we can form a tour $T$ by connecting the two paths such that $w(T) \geq w(H_1) + w(H_2) + w(a,b)$ without knowing the edge weights.
\end{lemma}

Our main techniques involve carefully stitching together greedy and random sub-tours, and establishing the tradeoffs between them. Our randomized algorithm returns two tours computed by two different sub-routines with equal probability: these are given by Algorithms \ref{alg_subroutine1_body} and \ref{alg_subroutine2_body}. 


\begin{algorithm}[hbtp]
\SetKwInOut{Input}{input}\SetKwInOut{Output}{output}
\Output{Tour $T_1$}
Let M be a greedy matching of size $k=\frac{N}{3}$, and $B$ be the nodes not in $M$\;
Complete $M$ using Lemma~\ref{lem_tourcompletion_body} to form a Hamiltonian path $H_T$ on nodes of $M$\;
Form a Hamiltonian path $H_B$ on $B$ using the following randomized algorithm.\;
\textbf{Randomized Path Algorithm} \;
Form a random permutation on the nodes in B\;
Join the nodes in the same order to form the path\;
(i.e., join the first and second nodes, second and third, and so on.)\;
\textbf{Final Output}
$T_1$ is the output formed by using Lemma~\ref{lem_pathtotour_body} for $H_1 = H_B$ and $H_2 = H_T$.
\caption{First Subroutine of the randomized algorithm for Max TSP}
\label{alg_subroutine1_body}
\end{algorithm}

\begin{lemma}
\label{lem_tsppart1_body}
The following is a lower bound on the weight of the tour returned by Algorithm~\ref{alg_subroutine1_body}
$$E[w(T_1)] \geq [\frac{3}{8} - \frac{3}{4N}]w(T^*) + \frac{6}{N}\sum_{x,y \in B}w(x,y) .$$
\end{lemma}

\begin{algorithm}[hbtp]
\SetKwInOut{Input}{input}\SetKwInOut{Output}{output}
\Output{Tour $T_2$}
Let M be a greedy matching of size $k=\frac{N}{3}$, and $B$ be the nodes not in $M$\;
Select $\frac{N}{6}$ edges uniformly at random from $M$\;
Complete these edges using Lemma~\ref{lem_tourcompletion_body} to form a Hamiltonian path $H_T$ with $\frac{N}{3}$ nodes\;
Let $A$ be the set of nodes in $M$ but not in $H_T$\;
\textbf{Randomized Alternating Path Algorithm}\;
Initialize $H_{AB} = \emptyset$\;
Select one node uniformly at random from $A$\;
Select one node uniformly at random from $B$\;
Add both the nodes to $H_{AB}$ in the same order\;
Remove them from $A$ and $B$ respectively \;
Repeat the above process until $A=B=\emptyset$\;
\textbf{Final Output}\;
$T_2$ is the output formed by using Lemma~\ref{lem_pathtotour_body} for $H_1 = H_{AB}$ and $H_2 = H_T$.
\caption{Second Subroutine of the randomized algorithm for Max TSP}
\label{alg_subroutine2_body}
\end{algorithm}

\begin{lemma}\label{lem_tsplast_body}
The expected weight of the tour returned by Algorithm~\ref{alg_subroutine2_body} is at least $[\frac{11}{16} - \frac{3}{4N}]w(T^*)- \frac{6}{N}\sum_{x,y \in B}w(x,y)$.
\end{lemma}

The final bound is obtained by using $E[w(T)] = \frac{1}{2}(E(w(T_1)] + E[w(T_2)])$.

\section{Conclusion}
In this paper we study ordinal algorithms, i.e., algorithms which are aware only of preference orderings instead of the hidden weights or utilities which generate such orderings. Perhaps surprisingly,  our results indicate that for many problems including Matching, $k$-sum clustering, Densest Subgraph, and Traveling Salesman, ordinal algorithms perform almost as well as algorithms which know the underlying metric weights, {\em even when the agents involved can lie about their preferences.} This indicates that for settings involving strategic agents where it is expensive, or impossible to obtain the true numerical weights or utilities, one can use ordinal mechanisms without much loss in welfare.

How do these algorithms stand in comparison to unconstrained ordinal algorithms that do not obey truthfulness? In the full version of this paper, we present non-truthful, ordinal algorithms for the same set of problems including a $4$-approximation algorithm for Densest subgraph and a $1.88$-approximation algorithm for Max TSP. In conjunction with the ordinal $1.6$-approximation algorithm for perfect matching from~\cite{anshelevichS16}, the improved approximation factors indicate a clear separation between the two classes of algorithms. On the surface, the improvement is not surprising since in many settings, truthfulness often places strong constraints on the set of allowed algorithms and techniques; indeed, all of our truthful mechanisms are derived using the three simple techniques outlined in Section~\ref{sec:prelim}. That said, given the absence of matching lower bounds in this work, the resolution of the gap between these two classes of algorithms is perhaps the most important question that is yet to be addressed.

\textbf{Acknowledgements} This work was supported in part by NSF awards CCF-1527497 and CNS-1218374.


\bibliography{bibliography}
\bibliographystyle{plain}

\appendix
\newpage 

\section{Appendix: Proofs from Section 2}

\subsubsection*{Greedy Algorithms for Other Problems}
We now provide some intuition on why greedy approaches do not lead to truthful algorithms for any of the other problems that we study in this work. \\
\textbf{Max $k$-sum and Densest $k$-subgraph}
In any clustering problem, using a greedy algorithm (or even serial dictatorship for that matter) could result in agents underplaying their most preferred node if that node will anyway be chosen in a later round. As a concrete example, consider the densest $k$-subgraph problem with $k=4$, and an instance with $6$ nodes whose preferences we define partially: $a$'s top $3$ nodes are $b,c,d$; $b$'s top two nodes are $a$ and $d$; $c$'s first two nodes are $a,e$; $d$'s top two preferences are $b$ and $e$ and finally, $e,f$ prefer each other as a first choice. 

Consider the simple algorithm that first picks a matching $M$ with $\frac{k}{2}$-edges and returns the same set of nodes as in the matching. Moreover, suppose that the algorithm's tie-breaking involves selecting edges containing and $a$ or $b$ before edges containing $e,f$ and then $c,d$. Now, under these preferences, we claim that node $a$ stands to improve her utility by lying when all the other agents are being truthful. To see why, first observe that if node $a$ truthfully reports her preferences, the algorithm returns $\{a,b,e,f\}$ as the solution set. On the other hand, if $a$ lies and points to $c$ as her first preference, then the algorithm picks $(a,c)$ first followed by $(b,d)$ resulting in the set $\{a,b,c,d\}$, which is strictly preferable from $a$'s perspective. A similar example holds for Max $k$-sum with fixed tie-breaking rules.

\textbf{Max TSP}
The negative example is a bit more subtle for Max TSP. Suppose that the greedy algorithm works by repeatedly picking undominated edges to build a forest of paths (so only edges that do not violate this property are maintained). Remember that an agent's utility for this problem is simply the sum of weights of the two edges that she is connected to. We construct a specific sub-instance where a node stands to gain by lying about her first preferences in order to increase her aggregate utility. It is not particularly hard to design a full set of preferences consistent with the sub-instance. 

Now, suppose that for a certain instance, the algorithm has already proceeded for a given number of rounds resulting in a forest of three disjoint paths: $\{a_1, a_2, a_3\}$, $\{b_1,b_2,b_3\}$, and $\{c_1, c_2, c_3\}$. Our antagonist-in-chief for this instance will be a separate node $x$ whose first four choices are $a_1 > a_3 > b_1 > t$, with $w(x,t) = 1$ and $w(x,a_1) = w(x,a_3) = w(x,b_1) = 2$. Moreover, suppose that $a_1$'s first and second choices are $x > c_1$, for $c_1$, it is $a_1 > b_1$, for $b_1$: $c_1 > x$, and finally, $a_3$'s top choice is $x$. Now, if $x$ is truthful, the unfolding series of events among these nodes will result in the addition of the edges $(x,a_1)$ and $(x,t)$ giving $x$ a total utility of $3$. On the other hand, it is preferable for $x$ to make $a_3$ its first preference resulting in $x$'s two edges being $(x,a_3)$ and $(x,b_1)$, which is a strictly better solution from $x$'s perspective. $\qed$

\section{Appendix: Proofs from Section 3: $1.7638$-Approximation Algorithm for Weighted Perfect Matching}
\label{app:perfectMatchingTruthful}

\subsection{Useful but Generic Lemmas}
We begin with some general lemmas that do not bear any obvious relation to greedy or random matchings, but will be useful in proving some of our results later on.

\begin{lemma}
\label{lem_genericffunctiondiffernetial}
Consider the following function of two variables $(x, \alpha)$ whose domains are as follows: $x \in [0, 0.5]$ and $\alpha \in [0,1]$.

$$f(x,\alpha) = \alpha(\frac{1}{2} - x) + (1-\alpha)(1 - x - \frac{x}{\alpha}).$$

For any fixed $x \in [0, \frac{1}{8}]$, $f(x)$ is not increasing from $\alpha = \frac{1}{2}$ to $1$. That is, as long as $x \in [0,\frac{1}{8}]$,

$$\max_{\alpha\in[0.5,1]} f(x,\alpha) = f(x,\frac{1}{2}).$$

\end{lemma}
\begin{proof}
For a fixed $x$, we can differentiate $f$ with respect to $\alpha$, and get

$$\frac{\partial f}{\partial \alpha} = \frac{x}{\alpha^2} - \frac{1}{2}.$$

The lemma follows from the observation that the derivative is not positive when $\alpha^2 \geq 2x$. $\hfill\qed$
\end{proof}

Our next set of lemmas allow us to establish upper bounds on dot-products of \emph{weight vectors}. For better understanding, one can imagine these vectors to be the weights of the edges in a greedy matching. Specifically, the lemmas help us identify the distribution of the weights that lead to our worst case bounds. We begin with the following trivial lemma.

\begin{lemma}
\label{lem_generic_weights1}
Consider two vectors $(w^1_i)_{i=1}^n$ and $(w^2_i)_{i=1}^n$ which are identical (i.e., $w^1_i=w^2_i$), except that $\exists r_1 \leq r_2$, $\epsilon > 0$ such that $w^1_{r_1} = w^2_{r_1} + \epsilon$ and $w^1_{r_2} = w^2_{r_2} - \epsilon$.

Let $\vec{a}$ be any fixed vector of the same length satisfying $a_1 \geq a_2 \geq \ldots \geq a_n$. Then,

$$\sum_{i=1}^n w^1_i a_i \geq \sum_{i=1}^n w^2_ia_i.$$
\end{lemma}
By repeatedly applying the above lemma one can identify weight vectors that dominate all other weight vectors with respect to the above sum. In essence, the above lemma indicates that it is always preferable (higher dot product) to transfer the weights from the larger indices of a vector to smaller indices. By repeatedly using the lemma, one arrives at the following corollary.

\begin{corollary}
\label{corr_generic_weights1}
Consider two vectors $(w^1_i)_{i=1}^n$ and $(w^2_i)_{i=1}^n$ that satisfy the following conditions
\begin{enumerate}

\item $\sum_{i=1}^n w^1_i = \sum_{i=1}^n w^2_i$

\item $\exists$ some index $k$ such that for every $r \leq k$, $w^1_r \geq w^2_r$ and for every $r > k$, $w^1_r \leq w^2_r$.
\end{enumerate}

Let $\vec{a}$ be any fixed vector of the same length satisfying $a_1 \geq a_2 \geq \ldots \geq a_n$. Then,
$$\sum_{i=1}^n w^1_i a_i \geq \sum_{i=1}^n w^2_ia_i.$$

\end{corollary}

Our final lemma is somewhat more specific and identifies a certain condition under which it is useful to transfer weight from the lower indices to the higher indices.

\begin{lemma}
\label{lem_genericweights2}
Consider two vectors $(w^1_i)_{i=1}^{2n+1}$ and $(w^2_i)_{i=1}^{2n+1}$ that satisfy the following conditions
\begin{enumerate}

\item $\sum_{i=1}^n w^1_i = \sum_{i=1}^n w^2_i$, and $\sum_{i=n+1}^{2n+1} w^1_i = \sum_{i=n+1}^{2n+1} w^2_i$

\item $\sum_{i=1}^n w^2_i = \sum_{i=n+1}^{2n+1}w^2_i$.

\item $w^2_1 = w_0$; for every $2 \leq i \leq 2n$, $w^2_i = \bar{w} \leq w_0$, and $w^2_{2n+1} = w_f$.

\item for every $1 \leq i \leq 2n$, $w^1_i = \bar{w}+\epsilon$ for some $\epsilon > 0$ and $w^1_{2n+1} = 0$.
\end{enumerate}
Then, for any $N \geq 2n+1$, $$\sum_{i=1}^{2n+1} w^1_i (N - 2i) \geq \sum_{i=1}^{2n+1} w^2_i(N - 2i).$$

\end{lemma}
\begin{proof}
Let us begin by establishing exact relationships between $w_0, \bar{w}, w_f, \epsilon$. First, we exploit Condition $(2)$ to show that $w_0 - \bar{w} = w_f$. Substituting the values of $w^2_i$ in Condition (2) gives us,

$$w_0 + \bar{w}(n-1) = \bar{w}n + w_f.$$

Adding and subtracting $\bar{w}$, from the LHS gives us the desired equation. Next, we derive an exact formula for $\epsilon$ in terms of $w_0, \bar{w}, w_f$, by adding up the two halves of Condition $(1)$.

$$w_0 + (2n-1)\bar{w} + w_f = 2n(\bar{w} + \epsilon).$$

Therefore, $\epsilon = \frac{w_0 - \bar{w} + w_f}{2n}.$ Now, we can show that the lemma follows from straightforward algebra. Consider the difference between the RHS and LHS of the lemma claim, this is given by,

\begin{align*}
RHS - LHS & =  (\bar{w} + \epsilon - w_0)(N-2) + \sum_{i=2}^{2n}\epsilon(N-2i) - w_f(N-2(2n+1)\\
& = \sum_{i=1}^{2n}\epsilon (N-2i) - \{(w_0 - \bar{w})(N-2) + w_f(N-2(2n+1))\}\\
& = \epsilon (2nN - 2n(2n+1)) - \{(w_0 - \bar{w})(N-2) + w_f(N-2(2n+1))\}\\
& = (w_0 - \bar{w} + w_f)(N- 2n - 1) - \{(w_0 - \bar{w})(N-2) + w_f(N-2(2n+1))\} \\
& = w_f (2n + 1) - (w_0 - \bar{w})(2n - 1)
\end{align*}

Substituting $w_0 - \bar{w} = w_f$ immediately tells us that $RHS - LHS \geq 0$, and thus we complete the lemma. $\hfill\qed$
\end{proof}

\subsection*{General Properties of Matchings}

Here we restate a simple lemma from \cite{anshelevichS16}. Since the proof is quite easy, we re-state the full proof here for completeness.

\begin{lemma}
\label{lem_matchingupper}
(Upper Bound) Let $G=(T,E)$ be a complete subgraph on the set of nodes $T \subseteq \mathcal{S}$ with $|T|=n$, and let $M$ be any perfect matching on the larger set $\mathcal{S}$. Then, the following is an upper bound on the weight of $M$,
$$w(M) \leq \frac{2}{n} \sum_{\substack{x \in T \\ y \in T}} w(x,y) + \frac{1}{n}\sum_{\substack{x \in T \\ y \in \mathcal{S} \setminus T}}w(x,y) $$
\end{lemma}

\begin{proof}
Fix an edge $e=(x,y) \in M$. Then, by the triangle inequality, the following must hold for every node $z \in T$: $w(x,z)+w(y,z) \geq w(x,y)$. Summing this up over all $z \in T$, we get
$$\sum_{z \in T}w(x,z) + w(y,z) \geq n w(x,y) = n(w_e).$$

Once again, repeating the above process over all $e \in M$, and then all $z \in T$ we have
$$nw(M) \leq 2 \sum_{\substack{x \in T \\ y \in T}} w(x,y) + \sum_{\substack{x \in T \\ y \in \mathcal{S} \setminus T}}w(x,y) $$

Each $(x,y) \in E$ appears twice in the RHS: once when we consider the edge in $M$ containing $x$, and once when we consider the edge with $y$.
\end{proof}

\subsection{Properties of Greedy Matchings}

\subsubsection*{Notation} We begin with some notation that helps us better characterize the solution returned by the greedy algorithm. Suppose that $GR$ denotes the output (set of edges) of the greedy algorithm (Algorithm \ref{alg_greedy}) for a given instance. Then, we use $w^{GR}_i$ to represent the weight of the $i^{th}$
largest edge in $GR$. For the rest of this proof, we will be abusing notation when expressing the total weight of edges inside a set. Specifically, $(i)$ for a set of edges $E$ (for e.g., $GR$), $w(E)$ ($w(GR)$) will denote the sum of the weights of edges inside $E$ (total weight of solution returned by greedy algorithm), $(ii)$ for a set of nodes $S$, $w(S)$ denotes the the total weight of the edges inside the induced graph on $S$, and $(iii)$ for disjoint $S,T \subseteq \mathcal{N}$, $w(S,T)$ is the total weight of the edges induced in the (complete) bipartite graph between $S,T$.

\subsubsection*{Top $x$-matching}
Given an instance, the corresponding greedy matching $GR$, and a fraction $x \leq 1$, we define the top $x$-matching $GR_x$, with respect to the given greedy matching to be the set of the top (maximum weight) $x$ fraction of edges inside $GR$. That is, if $|GR| = n$, $|GR_x| = xn$ and every edge in $GR_x$ has a weight no smaller than any edge in $GR \setminus GR_x$. Finally, let $\mathcal{N}(GR_x)$ denote the nodes that together make up $GR_x$.

Our first proposition highlights certain fundamental but very crucial properties that apply to all greedy matchings. These properties, especially local stability, will provide us with a much nicer platform towards a more detailed analysis of the greedy matching.

\begin{proposition}
\label{prop_greedyfundamental}
Suppose that $GR$ ($|GR| = n$) denotes the output of Algorithm~\ref{alg_greedy} for some instance $\mathcal{N}$, $GR' \subseteq GR$ and $T = \mathcal{N}(GR')$ is the set of nodes that form the edges in $GR'$.
\begin{enumerate}
\item \textbf{(Local Stability)} Suppose that Algorithm~\ref{alg_greedy} is run on the sub-instance consisting only of the nodes in $T$ with the same set of (sub) preferences. Then its output has to be $GR'$.

\item Suppose that $GR' = GR_x$ for some $x \leq 1$, and let $k=xn+1$. Then, in the induced subgraph on $\mathcal{N} \setminus T$, the maximum weight of any edge is at most $w^{GR}_k$.
\end{enumerate}
\end{proposition}
\begin{proof}
We prove both the properties by contradiction. Suppose that local stability does not hold and running the greedy algorithm on the subinstance returns  a matching $GR'' \neq GR'$. Consider $GR' = (e_1, e_2, \ldots, e_r)$, where the edges are ordered based on the position in which they were selected by Algorithm~\ref{alg_greedy} on the full instance, i.e., $e_1$ was first edge in $GR'$ to be selected by the algorithm, $e_2$ was the second edge and so on.

Denote by $1 \leq t \leq r$, the smallest index such that $e_t = (x,y) \notin GR''$. Without loss of generality, suppose that the greedy matching algorithm on the sub-instance picks an edge containing $x$ before it picks an edge containing $y$. Call this edge $(x,z)$; then we know exists some $z \neq y$, such that $(i)$ $(x,z) \in GR''$, and $(ii)$ $z \in e_{t_2}$ for some $t_2 > t$. Notice that by then definition of $t$, $(e_1, \ldots, e_{t-1}) \subseteq  GR''$ and so $z$ cannot belong to any of these edges if it is matched to $x$.

Consider the round in which the greedy algorithm on the sub-instance picked $(x,z)$: in this round $(x,z)$ is an undominated edge and $y$ is still available. Therefore, $x$ prefers $z$ to $y$. However, consider the round in which the original greedy algorithm picked $(x,y)$, clearly $(x,y)$ is once again undominated and $z$ was still avaiable in this round as $t_2 > t$. Therefore, with respect to this algorithm $x$ prefers $y$ to $z$, which is an overall contradiction since we assumed that the preferences between nodes in $T$ are not altered.

The second part of the lemma is easier to show. Assume by contradiction that $\exists$ an edge $(x^*, y^*)$ such that $w(x^*,y^*) > w^{GR}_k$ is the maximum weight edge in the given induced subgraph. A maximum weight edge is always undominated, and an undominated edge never ceases to be undominated, therefore $(x^*, y^*) \in GR$ and more specifically $(x^*, y^*) \in GR \setminus GR_x$. However, this violates the fact that $w^{GR}_k$ is the largest weight edge in $GR \setminus GR_x$. $\hfill\qed$
\end{proof}

Local stability has powerful consequences. For a given instance, we can take a subset of the greedy matching and show that all of the properties that apply to greedy matchings in general also apply to the subset, as it can be treated as an independent greedy matching. For the rest of this proof, we will treat greedy sub-matchings as independent greedy matchings on sub-instances, when it suits our needs.

The next lemma proves a simple but somewhat surprising fact. It is well-known that the greedy matching algorithm provides a $\frac{1}{2}$-approximation to the optimum matching. However, we show something much stronger: in order to get the same approximation factor, it is enough to consider only the heaviest $\frac{N}{4}$ edges and completely ignore the rest. This allows us to `further optimize' on the remaining edges using a random matching.

\begin{lemma}
\label{lem_greedy_tophalf}
Consider some instance $\mathcal{N}$: let $GR$ be the output of the greedy algorithm for this instance, and $OPT$ is the value of the maximum weight perfect matching for this instance. Then,

$$w(GR_\frac{1}{2}) \geq \frac{OPT}{2}.$$

\end{lemma}
\begin{proof}
We proceed via the standard charging argument applied to prove the half-optimality of the greedy algorithm. Pick any edge in $OPT$, say $e=(x,y)$, if $(x,y) \in GR$, we charge the edge to itself. Otherwise, at least one of $x$ or $y$ must be matched to an edge that yields it the same or better utility, i.e., w.l.o.g, $\exists (x,z) \in GR$ such that $w(x,z) \geq w(x,y)$. In this case, we charge $(x,y)$ to $(x,z)$. Clearly, every edge in $GR$ has anywhere between $0$ to $2$ edges (from $OPT$) assigned to it.

For any $e \in GR$, suppose that $s_e$ is the number of edges from $OPT$ assigned to $e$. By our charging argument, the following inequality must be true,

$$OPT \leq \sum_{i=1}^{\frac{N}{2}} w^{GR}_i s_{e_i},$$

where $e_i$ is the $i^{th}$ largest edge belonging to $GR$. Observe that $\sum_{i=1}^{\frac{N}{2}}s_{e_i} = \frac{N}{2}$. Consisder an alternative `slot vector', $(\vec{q})_{e_i}$ such that $q_{e_i} = 2$ if $i \leq \frac{N}{4}$ and $q_{e_i} = 0$ otherwise. As per Corollary~\ref{corr_generic_weights1}, we know that

$$\sum_{i=1}^{\frac{N}{2}}w^{GR}_i s_{e_i} \leq \sum_{i=1}^{\frac{N}{2}}w^{GR}_i q_{e_i} = 2w(GR_{\frac{1}{2}}).\hfill\qed$$

\end{proof}

\subsubsection*{Greedy Matchings: Induced Distances}
In the following series of lemmas, we prove upper bounds on the total weight of induced edges inside sets of nodes based on associated greedy matchings.

\begin{lemma}
\label{lem_gr_upperbound}
Let $T$ be some set of nodes with $|T|=n$ and let $GR$ denote the output of the greedy algorithm for this instance. Then,

$$w(T) \leq w(GR) + \sum_{i=1}^{n/2}2w^{GR}_i \{n-2i\}.$$
\end{lemma}

\begin{proof}
We know that $GR$ contains $\frac{n}{2}$ edges. Let $(x_i, y_i)$ denote the $i^{th}$ largest edge in $GR$, and for any $x$, let $\bar{T}_x$ denote the set of nodes not present in $GR_x$, i.e., $\bar{T}_x := \mathcal{N}(GR \setminus GR_x)$. For every $i \leq \frac{n}{2}$, we know from Proposition~\ref{prop_greedyfundamental} that $w^{GR}_i$ is the maximum weight edge in $\bar{T}_{\frac{2i}{n}}$. Therefore, for every such $i$, we know that $w^{GR}_i \geq w(x_i, j)$ and $w^{GR}_i \geq w(y_i, j)$ for all $j \in \bar{T}_{\frac{2i}{n}}$. Summing these up and adding a trivial inequality on both sides, we get

$$2|\bar{T}_{\frac{2i}{n}}|w^{GR}_i + w^{GR}_i = 2w^{GR}_i(n-2i) + w^{GR}_i \geq \sum_{j \in \bar{T}_{\frac{2i}{n}}}[w(x_i,j) + w(y_i,j)] + w(x_i, y_i).$$

Adding these up for all $i$ gives us the lemma. $\hfill\qed$

\end{proof}

\begin{lemma}
\label{lem_gr_crossupperbound}
Suppose that $GR$ denotes the output of the greedy matching algorithm for an instance $\mathcal{N}$. For some given $x$, define $T := \mathcal{N}(GR_x)$ and $\bar{T} = \mathcal{N} \setminus T$. Then,

$$w(T, \bar{T}) \leq 2w(GR_x) (|\bar{T}_x|).$$
\end{lemma}
The proof is very similar to that of the previous lemma, so we do not go over it again.

\begin{lemma}
\label{lem_grupperequalweight}
Consider some set of nodes $T$ with $|T|=n$ and let $GR$ denote the output of the greedy algorithm for this instance. Moreover, suppose that the maximum weight edge in $GR$, $w^{GR}_1\leq w^*$. Then,

$$w(T) \leq 2w(GR)\{n - \frac{w(GR)}{w^*}\}.$$


\end{lemma}
\begin{proof}
As a consequence of Lemma~\ref{lem_gr_upperbound}, we have that

$$w(T) \leq w(GR) + \sum_{i=1}^{n/2}2w^{GR}_i \{n-2i\}.$$

We know that for all $i$, $w^{GR}_i \leq w^*$.  Let $t=\lfloor\frac{w(GR)}{w^*}\rfloor$ be the maximum number of $w^*$ values that fit into $w(GR)$ and $r=w(GR)-tw^*$ be the remainder. Then, we can construct the following alternative weight vector: $w^1_i = w^*$ if $i \leq t$, $w^1_{t+1}=r$, and $w^1_i = 0$, otherwise. Using Lemma~\ref{lem_generic_weights1} repeatedly, we get a simplified inequality,

\begin{eqnarray*}
\sum_{i=1}^{n/2}2w^{GR}_i \{n-2i\} \leq \sum_{i=1}^{\frac{n}{2}}2w^1_i (n-2i) = \sum_{i=1}^{t}2w^*(n-2i) + 2r(n-2(t+1)) =\\
= 2w^*(nt - t(t+1)) + 2r(n-2(t+1)) = 2(w^*t+r)(n-t-1) - 2r(t+1)\\ \leq 2w(GR)(n-t-1) - 2(rt+r^2/w^*) = 2w(GR)(n-t-\frac{r}{w^*}-1) = \\ = 2w(GR)(n-\frac{w(GR)}{w^*}-1).
\end{eqnarray*}

We now wrap up the lemma,

$$w(T) \leq w(GR) + 2w(GR)(n - \frac{w(GR)}{w^*} - 1) \leq 2w(GR)(n-\frac{w(GR)}{w^*}). \hfill\qed$$
\end{proof}

Our next lemma is perhaps among the most crucial and technically involved of the lemmas in this matching proof.

\begin{lemma}
\label{lem_grtotaldist2parts}
Suppose that $GR$ denotes the output of the greedy algorithm for a given instance described by a set $T$ of nodes with $|T|=n$. Moreover, suppose that for a given $x$ in the range $[0,\frac{1}{4}]$, we have $w(GR_{2x}) \leq \frac{1}{2}w(GR)$. Then,

$$\frac{1}{n}w(T) \leq 2w(GR)\{1 - 2x\}.$$

\end{lemma}
\begin{proof}
From Lemma~\ref{lem_gr_upperbound}, we know that
$$w(T) \leq w(GR) + \sum_{i=1}^{n/2}2w^{GR}_i \{n-2i\}.$$

Our goal for this lemma is to show that (over all possible distributions of greedy edge weights satisfying the condition in the lemma), the maximum value of the second term in the above inequality is obtained when the top $2xn$ edges of the greedy matching all have the same weight, specifically $\frac{w(GR)}{2xn}$. First, define $m = xn$, i.e., $w_m^{GR}$ is the weight of the smallest edge in $GR_{2x}$ (since $GR$ has $n/2$ edges, $GR_{2x}$ will have $xn$ edges). Our proof will crucially depend on the weight of $m^{th}$ heaviest edge in $GR$, so let us use $\bar{w}$ to denote $w^{GR}_m$. We begin with some less ambitious sub-claims before showing the main result.

$$ \text{(Sub-Claim 1)} \sum_{i=1}^{m}2w^{GR}_i \{n-2i\} \leq 2w_0(n-2) + \sum_{i=2}^m 2\bar{w}(n-2i),$$

where $w_0$ is defined so that $w_0 + \sum_{i=2}^m \bar{w} = \sum_{i=1}^m w^{GR}_i$. To see why this is the case, consider the two equal-length vectors $\vec{w^1} = (w_0, \bar{w}, \ldots, \bar{w})$ and $\vec{w^2} = (w^{GR}_1, \ldots, w^{GR}_m)$. Since every entity in $\vec{w^2}$ is at least $\bar{w}$ (by definition) and the two vectors sum up to the same quantity, it must be the case that $w_0 \geq \bar{w}$. So, we can apply our general Corollary~\ref{corr_generic_weights1} with $k=1$ and get the sub-claim. Next, we state a similar claim for the second half of the edges in $GR$.

$$ \text{(Sub-Claim 2)} \sum_{i=m+1}^{\frac{n}{2}}2w^{GR}_i \{n-2i\} \leq \sum_{i=m+1}^{2m} 2\bar{w}(n-2i) + 2w_{f}(n - 2(2m+1)),$$

where $w_f$ is defined so that $\sum_{i=m+1}^{2m+1} \bar{w} + w_f = \sum_{i=m+1}^{\frac{n}{2}} w^{GR}_i$. As per the lemma statement, we have that $\sum_{i=m+1}^{2m+1} \bar{w} + w_f = w(GR)-w(GR_{2x}) \geq w(GR_{2x})= \sum_{i=1}^m w^{GR}_i = w_0 + \sum_{i=2}^m \bar{w} \geq m\bar{w}$. The proof of Sub-Claim $2$ comes from an easy (repeated) application of Lemma~\ref{lem_generic_weights1} since for every $i > m$, $w^{GR}_i \leq \bar{w}$, so we are simply transferring the weights to the smaller indices. Combining both the sub-claims, we get an intermediate inequality, from which it is more convenient to arrive at the lemma.

$$\sum_{i=1}^{n/2}2w^{GR}_i \{n-2i\} \leq 2w_0(n-2) + 2\sum_{i=2}^{2m}\bar{w}(n-2i) + 2w_f (n - 2(2m+1)).$$

Without loss of generality, we assume that $w_0$ and $w_f$ are defined so that $w_0 + \sum_{i=1}^{m} \bar{w} = \sum_{i=m+1}^{2m}\bar{w} + w_f$ or else we can always transfer some weight from $w_f$ to $w_0$, which only leads to an increase in the right hand side of the above inequality. This is equivalent to saying that the worst case for our lemma is when $w(GR_{2x}) = \frac{1}{2}OPT$.

If $w_0 = \bar{w}$, then $w_f = 0$, and we are done. So suppose that $w_0 > \bar{w}$. The following sub-claim completes our proof for the $x \leq \frac{1}{4}$ case.

$$\text{(Sub-Claim 3)}~ 2w_0(n-2) + 2\sum_{i=2}^{2m}\bar{w}(n-2i) + 2w_f (n - 2(2m+1))  \leq \sum_{i=1}^{2m} 2(\bar{w}+\epsilon)(n-2i),$$

where $\epsilon > 0$ is defined correspondingly in order to maintain the aggregate weight.  The sub-claim is directly derived from Lemma~\ref{lem_genericweights2}.

Finally, we plug all of these into our original generic bound to get

$$w(T) \leq w(GR) + 2\sum_{i=1}^{2m}(\bar{w} + \epsilon) \{n-2i\} \leq 2nw(GR)(1-2x). \hfill\qed $$

\end{proof}

\subsection{Lemmas concerning Random Matchings}
Having carefully laid down the foundation for a careful future analysis of greedy matchings, we now move on to showing simple generic lower bounds for random matchings that are applicable across a variety of situations. We begin with an obvious proposition that sets the stage for more involved bounds.

\begin{proposition}
Suppose that $RD$ denotes the random matching for a given instance $\mathcal{N}$. Consider the following two ways of partitioning $\mathcal{N}$: $(i)$ $\mathcal{N} = T \cup B$ with disjoint $T,B$, and $(ii)$, $\mathcal{N} = T \cup B_1 \cup B_2,$ where the 3 sets are once again disjoint. Then,

\label{prop_rd_lowerboundgeneric1}
\begin{enumerate}
\item $E[w(RD)] = \frac{1}{N}\left(w(T) + w(T,B) + w(B)\right)$.

\item $E[w(RD)] = \frac{1}{N}\{w(T) + w(T,B_1) + w(B_1) + w(T\cup B_1,B_2) + w(B_2)\}$.
\end{enumerate}
\end{proposition}

\begin{lemma}
\label{lem_rand_lowerbound2}
Suppose that $RD$ denotes the random matching for a given instance $\mathcal{N}$. Consider the following two ways of partitioning $\mathcal{N}$: $(i)$ $\mathcal{N} = T \cup B$ with disjoint $T,B$, and $(ii)$, $\mathcal{N} = T \cup B_1 \cup B_2,$ where the 3 sets are once again disjoint. Then, for $OPT$ denoting the weight of the maximum-weight matching,

\begin{enumerate}
\item $E[w(RD)] \geq \frac{1}{N}\{w(T) + |B|OPT - w(B)\}.$

\item $E[w(RD)] \geq \frac{1}{N} \{ w(T) + w(T,B_1) + w(B_1) + |B_2|OPT - w(B_2)\}$

\end{enumerate}

\end{lemma}
\begin{proof}
The first statement of the lemma is obtained by applying Lemma~\ref{lem_matchingupper} with $w(M) = OPT$, $\mathcal{S} = \mathcal{N}$, $T = B$.

The second statement is obtained by applying Lemma~\ref{lem_matchingupper} with $w(M) = OPT$, $\mathcal{S} = \mathcal{N}$, $T = B_2$.

\end{proof}

Our final corollary is obtained by adding the two parts of Lemma~\ref{lem_rand_lowerbound2} and dividng by two.
\begin{corollary}
\label{corr_random_lower3}
Suppose that $RD$ denotes the random matching for a given instance $\mathcal{N}$. Consider the following two ways of partitioning $\mathcal{N}$: $(i)$ $\mathcal{N} = T \cup B$ with disjoint $T,B$, and $(ii)$, $\mathcal{N} = T \cup B_1 \cup B_2,$ where the 3 sets are once again disjoint. Moreover, suppose that $B_1, B_2 \subseteq B$. Then,

$$E[w(RD)] \geq \frac{1}{N}\{ w(T) + \frac{1}{2}(w(T,B_1) + (|B| + |B_2|)OPT - w(B_1,B_2)) - w(B_2) \}.$$

\end{corollary}

\section{Appendix: Proofs from Section \ref{sec:kmatching}: Max $k$-Matching}

\begin{thm_app}{thm_rsdmaxk}
Random serial dictatorship is a universally truthful mechanism that provides a $2$-approximation for the Max $k$-matching problem.
\end{thm_app}

\begin{proof}
\textbf{Notation:} Given any set of nodes $S$, we use $\bar{G}(S)$ to denote the directed first preference graph on $S$ defined as follows: for every $i \in S$, there is a directed edge from $i$ to its most preferred agent in $S - \{i \}$. Our algorithm could be viewed as selecting one edge from $\bar{G}(T)$ uniformly at random in each iteration, where $T$ denotes the set of available agents.

For any set $S$, define $S^{-1} \subset S$ to be the random set of nodes remaining in $S$ after removing one edge uniformly at random from $\bar{G}(S)$, i.e., $S^{-1} := S - \{i,j\}$ with probability $\frac{1}{|S|}$ for every $(i,j) \in \bar{G}(S)$. Finally, we define $OPT(S,r)$ to denote the (weight of the) maximum weight $r$-matching in $S$ (containing $r$ edges). When it is clear from the context, we will abuse notation and use $OPT(S,r)$ to denote the optimum $r$-matching itself (as opposed to its value).

Our proof depends on the following crucial structural claim: we show that for any set $S$, $OPT(S,r) - E[OPT(S^{-1}, r-1)]$ is at most twice the weight of an edge chosen uniformly at random from $\bar{G}(S)$. This recursive claim implies that if at all we end up selecting a sub-optimal edge, then this does not hurt our solution by much since $E[OPT(S^{-1}, r-1)]$ is bound to be large, and we apply the algorithm recursively on $S^{-1}$.

\begin{claim}
(Structural Claim) For any given set $S \subseteq \mathcal{N}$ and $r \leq \frac{|S|}{2}$, we have that
\label{clm_keyclm_rsd}
$$ OPT(S,r) \leq E[OPT(S^{-1},r-1)] + \frac{2}{|S|}\sum_{e \in \bar{G}(S)}w(e)$$
\end{claim}

In the above claim, the expectation is taken over the different realizations of $S^{-1}$. In words, the claim bounds the change in the optima using the `increase in profit' of our algorithm.
We first show how this claim can be used to complete the proof of Theorem~\ref{thm_rsdmaxk}, and then detail the proof of Claim~\ref{clm_keyclm_rsd}.

\begin{proposition}
As long as Claim~\ref{clm_keyclm_rsd} is obeyed for every $S$, $r$, our algorithm provides a $2$-approximation to the Max $k$-matching.
\end{proposition}
\begin{proof}
Suppose that the algorithm proceeds in rounds ($1$ to $k$) where in each round exactly one edge is selected from the first preference graph. Define $S_i$ to denote the random set of available nodes at the beginning of round $i$ ($S_1 = \mathcal{N}$, and is deterministic). Then taking expectation over Claim~\ref{clm_keyclm_rsd}, we get that for every $i \leq k$,

$$E_{S_i}[OPT(S_i, k-i+1)] - E_{S_{i+1}}[OPT(S_{i+1},k-i)] \leq E_{S_i}[\frac{2}{|S_i|}\sum_{e \in \bar{G}(S_i)}w(e)]. $$

Moreover, if we define $Alg_i$ to denote the expected weight of chosen edge in round $i$, the term in the RHS is simply twice $Alg_i$. Therefore, we can simplify the above inequality as follows.

\begin{equation}
\label{eqn_rsd_intermed}
E_{S_i}[OPT(S_i, k-i+1)] - E_{S_{i+1}}[OPT(S_{i+1},k-i)] \leq 2Alg_i.
\end{equation}

We also know that $OPT(\mathcal{N}, k)$ can be written as a telescoping summation, $OPT(\mathcal{N} ,k ) =  \sum_{i=1}^k E[OPT(S_i, k-i+1)] - E[OPT(S_{i+1}, k-i)]$. After bounding the terms in the right hand side of the summation using Equation~\ref{eqn_rsd_intermed}, we complete the proof,

$$OPT(\mathcal{N},k) - E[OPT(S_{k+1},0)] \leq \sum_{i=1}^k 2Alg_i.$$

Since $E[OPT(S_{k+1},0)] =0$, and the RHS of the above algorithm is exactly the expected weight of the solution returned by our algorithm, the proposition follows. $\hfill\qed$

\end{proof}

It only remains to prove Claim~\ref{clm_keyclm_rsd}, which we complete now.

\subsubsection*{Proof of Claim~\ref{clm_keyclm_rsd}} We need to prove that $OPT(S,r) \leq E[OPT(S^{-1},r-1)] + \frac{2}{|S|}\sum_{e \in \bar{G}(S)}w(e)$. Now, for any $i \in S$, we use $o_i$ to denote the agent $i$ is matched to in $OPT(S,r)$. If the agent is unmatched in $OPT(S,r)$, we let $o_i$ be a null element. We also extend the notion of edge weights so that $w(i,\emptyset) = 0$ for all $i$. Finally, given any $i \in S$, let $s_i$ denote $i$'s most preferred node in $S$, i.e., the node to which $i$ has an outgoing edge in $\bar{G}(S)$.

Suppose that the (random) serial dictatorship removes the edge $(a,s_a)$ from $\bar{G}(S)$. We proceed in two cases based on whether or not $a$ is matched to a non-null agent in $OPT(S,r)$. Let $E_1$ denote the subset of edges in $\bar{G}(S)$ such that $a$ is matched to an actual agent in $OPT$, i.e., $o_a \neq \emptyset$. Note that $s_a$ may or may not be matched in $OPT$. Then, for any $(a,s_a) \in E_1$, we have that

$$OPT(S - \{a,s_a\},r-1)) \geq OPT(S,r) - w(a,o_a) - w(s_a,o_{s_a}) + w(o_a,o_{s_a}).$$

That is, $OPT(S - \{a,s_a\},r-1)$ is at least as good as the matching obtained by pairing up $o_a,o_{s_a}$ and leaving the other edges of $OPT(S,r)$. The above inequality is robust to $o_{s_a}$ being empty.

Observe that by definition $w(a,o_a) \leq w(a,s_a)$ and via the triangle inequality, $w(s_a,o_{s_a}) \leq w(o_a,s_a) + w(o_a, o_{s_a}) \leq w(o_a,s_{o_a}) + w(o_a,o_{s_a})$. So, we get a simplified charging argument for edges in $E_1$,

\begin{equation}
\label{eqn_rsde1}
OPT(S - \{a,s_a\},r-1)) \geq OPT(S,r) - w(a,s_a) - w(o_a, s_{o_a}).
\end{equation}

Finally, let us denote the remaining edges in $\bar{G}(S)$ as $E_2$, for every edge in $E_2$; we know that $o_a = \emptyset$, $o_{s_a}$ may or may not be empty. We claim that for both of these cases

\begin{equation}
\label{rqn_rsde3}
OPT(S - \{a,s_a\},r-1)) \geq OPT(S,r) - 2w(a,s_a).
\end{equation}

The main idea in this case is provided by Lemma~\ref{lem_rsdedgebound}, from which we infer that the weight of any edge in $OPT(S,r)$ is at most twice $w(a,s_a)$. So, if $o_{s_a} = \emptyset$, then $OPT(S - \{a,s_a\},r-1))$ is at least $OPT(S,r) - w(a,o_a)$. In the case that both of them are null, one can simply remove any one edge from $OPT(S,r)$ to get a lower bound for $OPT(S - \{a,s_a\},r-1))$.

We are now ready to complete the proof.

\begin{align*}
E[OPT(S^{-1},r-1)] = & \frac{1}{|S|}\sum_{(a,s_a) \in \bar{G}(S)}OPT(S-\{a,s_a\},r-1) \\
& \geq OPT(S,r) - \frac{1}{|S|}\{\sum_{(a,s_a) \in E_1}w(a,s_a) + w(o_a, s_{o_a}) + \sum_{(a,s_a) \in E_2}2w(a,s_a)\}  \\
& \geq OPT(S,r) - \frac{1}{|S|}\sum_{(a,s_a) \in \bar{G}(S)}2w(a,s_a) \\
\end{align*}

The crucial observation that leads us from line $2$ to line $3$ is that for any $(a,s_a) \in E_1$, the edge $(o_a, s_{o_a})$ must also belong to $E_1$. Therefore, in both of these cases, we are only counting $w(o_a, s_{o_a})$ twice.  $\hfill\qed$
\end{proof}

\section{Densest $k$-subgraph}\label{app:dks}

\begin{claim} (Trivial Algorithm)
Suppose that $k \geq \frac{N}{2}$, then the algorithm that selects a set of size $k$ uniformly at random from $\mathcal{N}$ is a universally truthful $6$-approximation algorithm for Densest $k$-subgraph.
\end{claim}
\begin{proof}
The truthfulness of this algorithm is quite obvious so we show the approximation factor. A trival upper bound for $OPT$ is $OPT \leq w(\mathcal{N})$, where $w(T)$ denotes the total weight of the graph induced by $T$.

Let $S$ be the random set returned by the above algorithm. Then, for some $i,j \in \mathcal{N}$, what is the probability that $i,j \in S$: this probability is exactly $\frac{{N-2 \choose k-2}}{{N \choose k}}$. As expected, the worst case occurs when $k=\frac{N}{2}$, giving us $Pr(i,j \in S) \geq \frac{N/2 - 1}{2(N - 1)} \geq \frac{1}{6}$ for $N >3$. Therefore, we have that $E[w(S)] \geq \frac{1}{6}\sum_{i,j \in \mathcal{N}}w(i,j)$. \hfill\qed
\end{proof}

\subsection*{General Lemmas}
Before proceeding with our main proof, we take a small detour and prove some generic lemmas that do not depend on our algorithm.

\begin{lemma}
\label{lem_dkskincreases}
For a given instance, let $OPT_r$ denote the weight of the densest subgraph of size $r$. Then,

$$OPT_{r+1} \leq OPT_r + \frac{2}{r-1}OPT_r.$$
\end{lemma}
\begin{proof}
Suppose that $O_{r+1}$ denotes the optimum solution to the Densest $r+1$-subgraph problem. Then, $OPT_{r+1} = \frac{1}{2}\sum_{i \in O_{r+1}}\sum_{j \in O_{r+1}}w(i,j)$. Then by the pigeonhole principle, there must exist at least one $i \in O_{r+1}$, such that $\sum_{j \in O_{r+1}}w(i,j) \leq 2\frac{OPT_{r+1}}{r+1}$: call this node $\tilde{i}$. Then, we have that

$$OPT_{r} \geq  w(O_{r+1} - \{\tilde{i}\}) = OPT_{r+1} - \sum_{j \in O_{r+1}}w(\tilde{i},j) \geq OPT_{r+1} - 2\frac{OPT_{r+1}}{r+1}.$$

The rightmost term is $OPT_{r+1}\frac{r-1}{r+1}$. Transposing the multiplicative factor gives us the lemma. \hfill$\qed$
\end{proof}

\begin{lemma}
\label{lem_rsdedgebound}
Consider any set of nodes $T$, and suppose that for some $x \in T$, $y$ denotes $x$'s most preferred node in $T$. Then for any given edge $(i,j) \in T \times T$, we have that $w(i,j) \leq 2w(x,y)$.
\end{lemma}
The lemma follows directly from an application of the triangle inequality. Specifically, $w(x,y)$ is at least half of the weight of the heaviest edge induced in $T$.

\begin{lemma}
\label{lem_pointtoset}
Let $T \subseteq \mathcal{N}$ be some set of agents and let $x$ be any given node. Then,

$$w(x,T) \geq \frac{1}{|T|-1}w(T).$$
\end{lemma}
\begin{proof}
The proof comes from the triangle inequality once again. For every edge $(i,j) \in T \times T$, we have that $w(i,j) \leq w(i,x) + w(j,x)$. Adding this inequality over all edges induced in $T$, we observe that for each $i \in T$, $w(i,x)$ appears in the RHS exactly $|T|-1$ times, i.e., there are $|T|-1$ edges inside $T$ containing $i$. The rest of the proof follows. \hfill$\qed$

\end{proof}

\section{Appendix: Proofs from Section \ref{sec:truthfulTSP}: Truthful Mechanism for Max Traveling Salesman Problem}

\begin{algorithm}[htbp]
{Initialize $T$ to be a random edge from the complete graph on $\mathcal{N}$\;
 Let $S$ be the set of available agents initialized to $\mathcal{N}$}\;
\While{$S \neq \emptyset$}{
pick one of the end-points of $T$, say $x$ \;
let $y$ denote $x$'s most preferred agent in $S$;
add $(x,y)$ to $T$ and remove $y$ from $S$\;
}
Complete $T$ to form a Hamiltonian cycle\;
\label{alg_tsp_app}
\caption{Path Building Serial Dictatorship for Max TSP}
\end{algorithm}

\begin{theorem}
\label{thm_maxtsp}
Algorithm~\ref{alg_tsp_app} is a universally truthful mechanism that provides a $2$-approximation to the optimum tour. Moreover, the algorithm provides a $(2+\epsilon)$-approximation to $OPT$, where $\epsilon \to 0$ as $N \to \infty$ even when the weights do not obey the metric assumption.
\end{theorem}
\begin{proof}
It is easy to see that this algorithms is truthful: when an agent $i$ is asked for its preferences, the first edge of $T$ incident to agent $i$ has already been decided, so $i$ cannot affect it. Thus, to form the second edge of $T$ incident to $i$, it may as well specify its most-preferred edge.

The proof proceeds via a straightforward paradigm where we charge edges in $T^*$, the welfare maximizing tour, to those in $T$, the solution returned by our algorithm. We first introduce some notation beginning with a simple tie-breaking rule that allows for convenient analysis. Specifically, suppose that $(a,b)$ denotes the first (random) edge added to $T$. Then, pick one of $a$ or $b$  (say $a$) uniformly at random, and term this node as the `dead node'. For the rest of algorithm, $a$ does not get to select another edge and remains as an end-point of $T$. The second edge containing $a$ is added only when the tour is completed to form a cycle. We remark that the randomization in the first step is {\em essential:} if we had selected the first edge based on the input preferences, then the first node could improve its utility by lying, and the algorithm would no longer be strategy-proof.

Next, for any $i \in \mathcal{N}$, we will use $t^*_1(i)$ and $t^*_2(i)$ to denote the two nodes that $i$ is connected to in $T^*$, and $t_1(i)$, $t_2(i)$ to the nodes connected to $i$ in $T$. Finally, suppose that $e_r$ denotes the random edge selected by the algorithm and $i_d$, the (random) dead node. In this proof, we show that for any realization of $e_r, i_d$, the optimum tour is at most twice the tour returned by our algorithm. Therefore, the same approximation bound also holds in expectation.

Fix some instantiation of $e_r, i_d$, call it $\tilde{e}_r$, $\tilde{i}_d$. Our charging argument comprises of two phases: in the first phase, we charge to the edges in $T$ all of the edges in $T^*$ except the ones containing the dead node $\tilde{i}_d$. While doing so, we ensure that for each edge in $T$, at most two edges in $T^*$ are charged to this edge. In the final phase, we carefully charge the edges in $T^*$ containing $\tilde{i}_d$ to certain edges in $T$ that were charged at most once in the first phase.

\subsubsection*{First Phase Charging}
Suppose that we use $S_i$ to denote the set of available nodes at the instant in our algorithm (for this particular instantiation of $e_r, i_d$) when an edge containing $i$ is added to $T$. The algorithm then proceeds to pick $i$'s most preferred agent in $S_i$ and adds the corresponding edge to $T$. Suppose that for every $i \in \mathcal{N}$, $t_2(i)$ denotes its most preferred node in $S_i$.

Now consider any edge $(x^*, y^*)$ in $T^*$ such that $x^*, y^* \neq \tilde{i}_d$. Suppose that $x^*$ was removed from the set of available nodes before $y^*$ during the course of the algorithm. Then, $y^* \in S_{x^*}$ and so, $w(x^*, t_2(x^*)) \geq w(x^*, y^*)$ and we can charge the edge $(x^*, y^*) \in T^*$ to $(x^*, t_2(x^*)) \in T$.  

After repeating this charging for every edge in $T^*$ except $(\tilde{i}_d, t^*_1(\tilde{i}_d))$, $(\tilde{i}_d, t^*_2(\tilde{i}_d))$, we end up with the following proposition.

\begin{proposition}
\label{prop_phase1tsp}
The following are true at the end of the first phase of charging.
\begin{enumerate}
\item At most two edges in $T^*$ are charged to any one edge in $T$.

\item No edges are charged to $(\tilde{i}_d, t_1(\tilde{i}_d))$, $(\tilde{i}_d, t_2(\tilde{i}_d)) \in T$.

\item At most one edge in $T^*$ is charged to any of the edges in $T$ containing $t^*_1(\tilde{i}_d)$, $t^*_2(\tilde{i}_d)$.
\end{enumerate}
\end{proposition}

\begin{proof}
\textbf{(Statement 1)} Consider any edge of the form $(i, t_2(i))$, as per our definitions, $i$ became unavailable before $t_2(i)$. Thus, the only edges charged to $(i,t_2(i))$ are those in $T^*$ containing $i$, and there can only be two such edges.

\noindent \textbf{Statement 2} Further, suppose that $(\tilde{i}_d, t_1(\tilde{i}_d))$ denotes the random edge in $T$. Clearly, edges in $T^*$ containing $t_1(\tilde{i}_d)$ are not charged to the random edge. Finally, no edge in $OPT$ is charged to $(\tilde{i}_d, t_2(\tilde{i}_d))$, since the latter node is the absolute last node to become unavailable.

\noindent \textbf{Statement 3} This is a direct consequence of the fact that we have not charged $(\tilde{i}_d, t^*_1(\tilde{i}_d))$, $(\tilde{i}_d, t^*_2(\tilde{i}_d))$.
\end{proof}

\subsubsection*{Second Phase Charging} We use the triangle inequality to charge the edge $(\tilde{i}_d, t^*_1(\tilde{i}_d))$:

$$w(\tilde{i}_d, t^*_1(\tilde{i}_d)) \leq w(\tilde{i}_d, t_2(\tilde{i}_d)) + w(t_2(\tilde{i}_d), t^*_1(\tilde{i}_d)) \leq w(\tilde{i}_d, t_2(\tilde{i}_d)) + w(t^*_1(\tilde{i}_d), t_2(t^*_1(\tilde{i}_d))).$$

The final inequality is due to the fact that $t_2(\tilde{i}_d) \in S_{t^*_1(\tilde{i}_d)}$ (since $t_2(\tilde{i}_d)$ is the absolute last node added to the tour, and so it is available during the entire runtime of the algorithm), and so $w(t_2(\tilde{i}_d), t^*_1(\tilde{i}_d)) \leq w(t^*_1(\tilde{i}_d), t_2(t^*_1(\tilde{i}_d)))$. Therefore, the edge $(\tilde{i}_d, t^*_1(\tilde{i}_d))$ can be charged to two edges in $T$, namely $(\tilde{i}_d, t_2(\tilde{i}_d))$ and $(t^*_1(\tilde{i}_d), t_2(t^*_1(\tilde{i}_d))$

Using exactly the same kind of argument, we can also charge the second edge containing $\tilde{i}_d$ in $T^*$ to two edges in $T$, namely $(\tilde{i}_d, t_2(\tilde{i}_d))$ and $w(t^*_2(\tilde{i}_d), t_2(t^*_2(\tilde{i}_d))$. This concludes the second phase of charging.

In conjunction with Proposition~\ref{prop_phase1tsp}, we have successfully charged every edge in $OPT$ by using at most two edges in $T$. This completes our two approximation.

\subsubsection*{Proof for the Non-Metric Case}
The main idea that leads to the bound for the non-metric case is that the first phase of charging does not use the metric nature of the weights in any way. Therefore, at the end of first phase, we charged all of the edges in $T^*$ minus the ones containing $\tilde{i}_d$ by using at most two edges in $T$. Note that this is for a particular instantiation.

Therefore, taking the expectation over every such instantiation, we get that

\begin{align*}
E_{i_d} [w(T^*) - w(i_d, t^*_1(i_d)) - w(i_d, t^*_2(i_d))] & \leq 2E[w(T)] \\
w(T^*) - \frac{2}{N}w(T^*) & \leq 2E[w(T)].
\end{align*}

For the second inequality, we used the fact that for any $i \in \mathcal{N}$, $Pr[i_d = i] = \frac{1}{N}$ and therefore every edge $(x^*, y^*)$ in $T^*$ appears with the negative sign twice: once when $x^*$ is dead, and once when $y^*$ is dead.

This completes the proof.
\end{proof}

\section{Appendix: Proofs from Section \ref{sec_bicriteria}}
\begin{theorem}
We can efficiently compute an ordinal $(\frac{4}{\beta^2}, \beta)$-approximate solution for the Densest $k$-subgraph problem for $\beta \leq 2$, i.e., a solution of size $\beta k$, whose value is at least $\frac{\beta^2}{4}$ times that of the optimum solution of size $k$.
\end{theorem}

\begin{proof}
The algorithm is described as follows: ``Let $M$ be a greedy matching of size $\beta \frac{k}{2}$. Return S, the set of nodes which are the endpoints of the edges in $M$".

The proof is somewhat complicated, and involves carefully charging different sets of node distances in $S^*$ to node distances in $S$. So, before giving the main proof, we provide a series of very general charging lemmas. We begin by defining a somewhat unusual `device' that guides our charging arguments. For the rest of this proof, given a matching $M$, we will use $N(M)$ to denote the set of nodes which form the endpoints of the edges in $M$.

Suppose that we are provided a matching $M$ of some given size, and a set $B \subseteq N(M)$. Now, given an integer $t \leq |B|$, define $M(t,B)$ to be the top (i.e., highest weight) $t$ edges in $M$, such each edge in $M(t,B)$ contains at least one node from $B$. We refer to $M(t,B)$ as the top-intersecting matching. In the following lemmas, we will highlight the versatility of the top-intersecting matching by charging different sets of inter node distances to this matching. Later, we will show this `device' can be used to prove the main theorem.

Note that in the lemmas that follow we will assume that $M$ is a greedy matching of size $k$, initialized with the complete edge set.

\begin{lemma}
\label{clm_bicrit2}
Suppose that $M$ is a greedy matching, and $B \subseteq N(M)$ is some given set of size $2m$. Then the following is a upper bound for the total distances of the edges inside $B$,

$$\sum_{x,y \in B}w(x,y) \leq \sum_{x,y \in N(M(m,B)) \cap B}w(x,y) + \frac{5r}{2}w(M(m,B)),$$

where $r = |B \setminus N(M(m,B))|$, i.e., $r$ is the number of nodes of $B$ that are not inside the set $N(M(m,B))$.
\end{lemma}

\begin{proof}
For convenience, let us use $A$ to denote the set $N(M(m,B)) \cap B$, i.e., the nodes of $B$ that are contained in $M(m,B)$. By definition of the top-intersecting matching, $|A| \geq m$, and therefore $r \leq m$, since $|A| + r = |B|$.

First, notice that the edges inside $B$ can be divided into three parts as follows with $A$ serving as the virtual partition.

$$\sum_{x,y \in A}w(x,y) + \sum_{\substack{x  \in A \\ y \in B \setminus A}}w(x,y) + \sum_{x,y \in B\setminus A}w(x,y).$$

The first term above is exactly the same as the first term in the RHS of the lemma statement. Therefore, it suffices if we show that the second term plus the third term above are at most $\frac{5r}{2}$ times the weight of the matching $M(m,B)$. We first consider edges going from $A$ to $B\setminus A$.

\textbf{First Part:} Suppose that $(x,y)$ is an edge where $x \in A$ and $y \in B \setminus A$. Let $(x,z)$ be the edge in $M(m,B)$ that contains $x$. Since the edges in $M$ were chosen in a greedy fashion and also, since $w(x,z)$ is at least as large as the edge in $M$ containing $y$ (by definition of $M(m,B)$) , we infer that $w(x,z) \geq w(x,y)$. In this fashion, we get that for a fixed $x \in A$, $\sum_{y \in B \setminus A}w(x,y) \leq rw(x,z)$, where $r$ is the number of nodes in $B\setminus A$.

 Summing up over all $x \in A$ and all $y \in B \setminus A$, we have $\sum_{x\in A, y \in B\setminus A}w(x,y) \leq 2r \times w(M(m,B)).$ Notice that for a given edge $(a,b) \in M(m,B)$, (at most) $2r$ edges in $A \times B\setminus A$ can be charged to this edge, which happens when both $a$ and $b$ belong to $A$. In summary, we have
$$\sum_{\substack{x \in A\\ y \in B\setminus A}}w(x,y) \leq 2r w(M(m,B)).$$

\textbf{Second Part:} Next, consider $B \setminus A$: let $M^*(B\setminus A)$ be the optimum matching using only the nodes in $B \setminus A$, and let $M(B \setminus A)$ be the (smallest) set of edges of $M$ containing the nodes in $B \setminus A$. Observe that the edges in $M(B \setminus A)$ do not belong to $M(m,B)$ by definition.

Since we chose the edges in $M$ in a greedy fashion, it is not hard to see that for every edge $e=(x,y)$ of $M^*(B\setminus A)$, either the edge $e$ is also in $M(B \setminus A)$, or one of the two edges of $M(B \setminus A)$ containing $x$ and $y$ must have higher weight than $e$. If this were not true, then $e$ would have been chosen for $M$, since it would be preferred by both $x$ and $y$. In particular, this means that the edge $e_{\max}$ which has the highest weight of all edges in $M(B \setminus A)$, must have weight at least that of any edge in $M^*(B\setminus A)$. Moreover, every edge of $M(m,B)$ has higher weight than $e_{\max}$, since otherwise $e_{\max}$ would have been chosen for $M(m,B)$ instead of the edges in it. Since $|M^*(B\setminus A)|=\frac{r}{2}$ and $|M(m,B)|=m\geq r$, we know that $M(m,B)$ has at least twice as many edges as $M^*(B\setminus A)$.
In conclusion, we have that $w(M^*(B\setminus A)) \leq \frac{1}{2}w(M(m,B))$. Finally, since the weight of a max-weight matching is at least that of a random matching, we know that $\sum_{x,y \in B\setminus A}w(x,y) \leq r \times w(M^*(B\setminus A))$, so we get

$$\sum_{x,y \in B\setminus A}w(x,y) \leq rw(M^*(B\setminus A)) \leq \frac{r}{2} w(M(m,B)).$$

Adding the upper bounds for the two parts, we get the lemma.
\end{proof}

\begin{lemma}
\label{clm_bicrit3}
Suppose that $M$ is a greedy matching, and suppose that $B$ and $C$ are two disjoint sets such that $B \subseteq N(M)$, and $C \cap N(M) = \emptyset$. Then the following is an upper bound for the edges going from $B$ to $C$

$$\sum_{x \in B, y \in C}w(x,y) \leq 2|C|w(M(m,B)),$$

where $|B|=2m$.
\end{lemma}

\begin{proof}
Suppose that $M(B)$ is the (minimal) set of edges in $M$ containing every node in $B$. Clearly, $M(B)$ contains at most $|B|=2m$ edges and encompasses $M(m,B)$.

Fix some $x \in B$: for every $y \in C$, we have that $w(x,y) \leq w(x,z)$, where the latter is the edge in $M$ containing $x$. This is because otherwise the edge $(x,z)$ would not have been undominated, and thus would not have been added to $M$. In this manner, we can charge $\sum_{x \in B, y \in C}w(x,y)$ to the matching $M(B)$, using at most $2|C|$ slots of each edge, and a total of $2m|C|$ slots. This means that we can use a slot transfer argument and transfer all the slots to the edges in $M(m,B)$, using at most $2|C|$ slots of each edge. Therefore, we conclude that
$$\sum_{x \in B, y \in C}w(x,y) \leq 2|C|w(M(m,B)).$$
\end{proof}

This bounds edges going from the intersecting part to the disjoint part.

\begin{lemma}
Suppose that $M$ is a greedy matching, and suppose that $B$ and $C$ are two disjoint sets such that $B \subseteq N(M)$, and $C \cap N(M) = \emptyset$. Then the following is an upper bound for the edges contained in $C$.

\label{clm_bicrit4}
$$\sum_{x ,y \in C}w(x,y) \leq \frac{(|C|^2)}{|M| - m}w(M \setminus M(m,B)),$$

where $|B|=2m$.
\end{lemma}

\begin{proof}
Suppose that $M^*(C)$ is the optimum matching containing only the edges in $C$. Since the weight of a max-weight matching is at least that of a random matching, we know that $\sum_{x ,y \in C}w(x,y) \leq |C|w(M^*(C))$.

Now, since the nodes in $C$ are not present in $M$, we know that for any edge $(x,y)$ in $M$ and $(a,b)\in M^*(C)$, it must be that $w(x,y)\geq w(x,a)$ and $w(x,y)\geq w(x,b)$. Applying the triangle inequality, we get that $w(a,b)\leq 2w(x,y)$. Thus, for every edge in $M$, its weight is at least half the weight of any edge in $M^*(C)$. Next, consider the set of $|M|-m$ edges in $M \setminus M(m,B)$: by the above argument we get that $w(M^*(C)) \leq \frac{|C|}{|M|-m}w(M \setminus M(m,B))$. Therefore, in conclusion, we get
$$\sum_{x ,y \in C}w(x,y) \leq |C|w(M^*(C)) \leq w(M\setminus M(m,B))\frac{|C|^2}{|M|-m}.$$
\end{proof}

Finally, we establish a lower bound on the distances of the edges inside $N(M)$ once again in terms of the top-intersecting matching.

\begin{lemma}
\label{clm_bicrit5}
Suppose that $M$ is a greedy matching. Given any $B \subseteq N(M)$ with $B=2m$, let $Top=N(M(m,B))$ and $A=B \cap Top$. Then the following is a piecewise lower bound for the edges inside $N(M)$.

\begin{enumerate}
\item $\sum_{x,y \in A}w(x,y) \geq \sum_{x,y \in A}w(x,y)$ (Trivial Lower Bound) \\

\item $\sum_{x \in Top, y\in Top\setminus A}w(x,y) + \sum_{x \in Top, y \in N(M)\setminus Top}w(x,y) \geq (2|M| - 2m + \frac{r}{2})w(M(m,B))$ \\

\item $\sum_{x,y \in N(M)\setminus Top}w(x,y) \geq (|M| -m)w(M \setminus M(m,B))$,
\end{enumerate}

where $r = |B \setminus A|$.
\end{lemma}

Note that when we say $x \in Top$ and $y \in Top\setminus A$, we are counting each edge only once. This is true for the other summations as well.

\begin{proof}
Recall that $Top$ refers to the set $N(M(m,B))$. Remember that $A$ is a subset of $Top$, and $B\setminus A$ has no node in common with $Top$.

The quantity $\sum_{x,y \in N(M)}w(x,y)$ includes (at least) `edges going from ($A$ to $A$), ($Top$ to $Top\setminus A$), ($Top$ to $N(M)\setminus Top$), and ($N(M)\setminus Top$ to $N(M)\setminus Top$).

\noindent\textbf{(Part 2.1):}

$$\sum_{\substack{x \in Top \\ y \in Top \setminus A}}w(x,y) \geq \frac{r}{2}w(M(m,B)).$$ Note that $|Top| = 2m$, and $|A| = (2m-r)$. Now, pick any edge $(x,y) \in M(m,B)$, and the $r$ nodes $z\in Top \setminus A$ and apply the triangle inequality. We get, $\sum_{z \in Top\setminus A}w(x,z) + w(y,z) \geq rw(x,y)$. Summing this up over all $(x,y) \in M(m,B)$, and all $z \in Top \setminus A$, and dividing by two (since we count some edges twice), we get the first statement.\\

\textbf{(Part 2.2)}

$$\sum_{\substack{x \in Top \\ y \in N(M)\setminus Top}}w(x,y) \geq (2|M| - 2m)w(M(m,B)).$$ This follows almost directly from taking each edge in $M(m,B)$, and every node in $N(M) \setminus B$, and applying the triangle equality.

Summing up Parts 2.1 and 2.2, gives us the second statement of our lemma.

\textbf{(Part 3)}

$$\sum_{x,y \in N(M)\setminus B}w(x,y) \geq (|M| -m)w(M \setminus M(m,B)).$$ This comes from applying Lemma~\ref{lem_matchingupper} to the set $N(M)\setminus B$.

\end{proof}

\subsubsection*{Applying the Generic Claims}

Now that we have some general properties of the top-intersecting matching, we shed light on how to apply the above framework towards our main theorem. Recall that our algorithm involves choosing a greedy matching $M$ of size $\frac{\beta k}{2}$, and using its endpoints to form the set $S$ of size $\beta k$. Let $S^*$ be the optimum solution to the Densest Subgraph problem for input parameter $k$. Define $B$ to be the set of nodes that are common to both $S$ and $S^*$, and let $|B|=2m$.

Now as per our definitions $M(m,B)$ is the set of $m$ highest weight edges in $M$ all of which contain at least one node from $B$, i.e., nodes from $S^*$. Given this framework, we can express the weight of our optimum solution $w(S^*)$ as $\sum_{x,y \in A}w(x,y) + \sum_{x\in B,y \in B\setminus A}w(x,y) + \sum_{x\in B,y\in S^*\setminus B}w(x,y) + \sum_{x,y \in S^* \setminus B}w(x,y).$ We take $A$ to be the subset of $B$ contained in $M(m,B)$, and $r=|B\setminus A|$.

Now, we are in a position to bound $w(S^*)$ in terms of $w(M(m,B))$, and $w(M\setminus M(m,B))$. Summing up Lemmas~\ref{clm_bicrit2},~\ref{clm_bicrit3},~\ref{clm_bicrit4} (take $C=S^*\setminus B$) , we get,

\begin{align*}
w(S^*) &  \leq \sum_{x,y \in A}w(x,y) & + [\frac{5r}{2} + 2(k-2m)]w(M(m,B)) \\
&& + \frac{2(k-2m)^2}{\beta k - 2m}w(M\setminus M(m,B)).
\end{align*}

Next, applying Lemma~\ref{clm_bicrit5}, we transform the above upper bound to a bound on $w(S)$. Once again, in order to avoid lengthy notation, we use $Top$ to refer to $N(M(m,B))$.

\begin{align*}
w(S^*) &\leq  \sum_{x,y \in A}w(x,y) \\
& + \frac{2(k-2m) + \frac{5r}{2}}{\beta k- 2m + \frac{r}{2}}[\sum_{\substack{x \in Top \\ y \in Top \setminus A}}w(x,y) +\sum_{\substack{x \in Top \\ y \in N(M)\setminus Top}}w(x,y)]\\ & + \frac{4(k-2m)^2}{(\beta k -2m)^2}\sum_{x,y \in N(M)\setminus Top}w(x,y) \\ & \leq \max\left(1,\frac{2(k-2m) + \frac{5r}{2}}{\beta k-2m + \frac{r}{2}}, \frac{4(k-2m)^2}{(\beta k - 2m)^2}\right)w(S)
\end{align*}

The rest of the proof is somewhat algebraic, so we only sketch the details here. Look at the second term inside the maximization function above. For a fixed $\beta$, and a fixed value of $m$, we can show using some simple calculus that the quantity $\frac{2(k-2m) + \frac{5r}{2}}{\beta k-2m + \frac{r}{2}}$ monotonically increases with $r$, and therefore is maximized when $r$ attains its maximum value. Moreover, we know that $r \leq m$, and $r \leq \beta k - 2m$, since $S$ only contains $\beta k$ nodes and $2m$ of them are inside $N(M(m,B))$. Now, consider $1 \leq \beta \leq \frac{6}{5}$, and take $r = \beta k - 2m$, we get
$$w(S^*) \leq \max\left(\frac{4k + 5\beta k - 18m}{3\beta k - 6m}, \frac{4(k-2m)^2}{(\beta k -2m)^2}\right)w(S)$$

Moreover, for a fixed $\beta$, both of the terms above reach their maximum value when $m$ is at its smallest ($m=0$) giving us $w(S^*) \leq \max\left(\frac{4 + 5\beta}{3\beta}, \frac{4}{\beta^2}\right)w(S)$ and so when $\beta \leq \frac{6}{5}$, we get $w(S^*) \leq \frac{4}{\beta^2}w(S)$.

Next, take $\beta \geq \frac{6}{5}$, and $r = m$ in our general bound for $w(S^*)$. We get,

$$w(S^*) \leq \max\left(\frac{2k-\frac{3m}{2}}{\beta k - \frac{3m}{2}}, \frac{4(k-2m)^2}{(\beta k -2m)^2}\right)w(S)$$

This time, for a fixed $\frac{6}{5} \leq \beta \leq 2$, the first term reaches its maximum value when $m$ is largest $(B=S^* \text{ or } 2m=k)$, and the second when $m$ is at its smallest $(m=0)$, giving us $w(S^*) \leq \max\left(\frac{2 - \frac{3}{4}}{\beta - \frac{3}{4}}, \frac{4}{\beta^2}\right)w(S)$ and so when $\beta \geq \frac{6}{5}$, we get $w(S^*) \leq \frac{4}{\beta^2}w(S)$.

So in both cases, we get our bound of $w(S^*) \leq \frac{4}{\beta^2} w(S)$, which completes the proof of the theorem.
\end{proof}

\section{Appendix: Proofs from Section \ref{sec:TSP}: 1.88 Approximation for Max TSP}
Before defining our randomized algorithm, we first present the following lemma, which gives a relationship between matching and Hamiltonian paths.

\begin{lemma}
\label{lem_tourcompletion}
Given any matching $M$ with $k$ edges, there exists an efficient ordinal algorithm that computes a Hamiltonian path $Q$ containing $M$ such that the weight of the Hamiltonian path in expectation is at least
$$[\frac{3}{2} - \frac{1}{k}]w(M).$$
\end{lemma}

\begin{proof}
We first provide the algorithm, followed by its analysis. Suppose that $K$ is the set of nodes contained in $M$.
\begin{enumerate}
\item Select a node $i \in K$ uniformly at random. Suppose that $e(i)$ is the edge in $M$ containing $i$.
\item Initialize $Q=M$.
\item Order the edges in $M$ arbitrarily into $(e_1, e_2, \ldots, e_k)$ with the constraint that $e_1=e(i)$.
\item For $j=2$ to $k$,
\item Let $x$ be a node in $e_{j-1}$ having degree one in $Q$ (if $j=2$ choose $x \neq i$) and $e_j=(y,z)$.
\item If $y >_x z$, add $(x,y)$ to $Q$, else add $(x,z)$ to $Q$.
\end{enumerate}

Suppose that $Q(j)$ consists of the set of nodes in $Q$ for a given value of $j$ (at the end of that iteration of the algorithm). We claim that $w(Q) \geq \frac{3}{2}w(M) - w(e_1)$, which we prove using the following inductive hypothesis,
$$w(Q(j)) \geq w(M) + \sum_{r=2}^j \frac{1}{2}w(e_r)$$

Consider the base case when $j=2$. Suppose that $e_1=(i,a)$, and $e_2=(x,y)$. Without loss of generality, suppose that $a$ prefers $x$ to $y$, then in that iteration, we add $(a,x)$ to $Q$. By the triangle inequality, we also know that $w(x,y) \leq w(x,a) + w(y,a) \leq 2w(x,a)$. Therefore, at the end of that iteration, we have
\begin{equation}
\label{eqn_temptspbb}
w(Q) = w(M) + w(x,a) \geq w(M) + \frac{1}{2}w(e_2).
\end{equation}

The inductive step follows similarly. For some value of $j$, let $x$ be the degree one node in $e_{j-1}$, and $e_j=(y,z)$. Suppose that $x$ prefers $y$ to $z$, then using the same argument as above, we know that $(x,y)$ is added to our desired set and that $w(x,y) \geq \frac{1}{2}w(y,z)$. The claim follows in an almost similar fashion to Equation~\ref{eqn_temptspbb} and the inductive hypothesis.

In conclusion, the total weight of the path is $\frac{3}{2}w(M) - w(e_1)$. Since the first node $i$ is chosen uniformly at random, every edge in $M$ has an equal probability ($p=\frac{1}{k}$) of being $e_1$. So, in expectation, the weight of the tour is $\frac{3}{2}w(M) - \frac{1}{k}w(M)$, which completes the lemma.\hfill$\qed$
\end{proof}

Our randomized algorithm involves returning two tours computed by two different sub-routines with equal probability. We define some pertinent notation before describing the two sub-routines. Suppose that $M$ is the solution returned by the Greedy Matching Procedure for $k=\frac{1}{3}N$, i.e., $M$ contains $\frac{2}{3}$ times the number of edges in any perfect matching. We know from Lemma 2.2 in \cite{anshelevichS16} that $w(M) \geq \frac{w(M^*)}{2}$, where $M^*$ is the optimum perfect matching in $\mathcal{N}$; in fact this is not difficult to see directly by classic charging arguments comparing greedy matchings to maximum-weight matchings.

Let $Top$ be the set of nodes whose edges form $M$, and let $B = \mathcal{N} \setminus Top$. Finally, given any Hamiltonian path $H$, we use $H(f)$ and $H(l)$ to denote the dangling nodes of $H$, i.e., the two endpoints of $H$ whose degrees are one. Before showing the algorithm, we give a simple lemma that provides a `nice way' to form a tour using two Hamiltonian Paths.

\begin{lemma}
\label{lem_pathtotour}
Let $H_1$ and $H_2$ be two Hamiltonian paths on two different sets of nodes. Then, we can form a tour $T$ by connecting the two paths such that $w(T) \geq w(H_1) + w(H_2) + w(H_1(f),H_1(l))$ without knowing the edge weights.
\end{lemma}

\begin{proof}
For ease of notation, we refer to the $4$ endpoints of the two paths in the following manner, $H_1(f) = a, H_1(l) = b, H_2(f) = x, H_2(l) = y$. Without loss of generality, suppose that $a$ prefers $x$ to $y$. Then, let $T=H_1 \cup H_2  \cup \left\lbrace (a,x) , (b,y) \right\rbrace$. In this case, we have that $w(a,b) \leq w(a,y) + w(b,y) \leq w(a,x) + w(b,y).$ Therefore, we get $w(T) = w(H_1) + w(H_2) + w(a,x) + w(b,y) \geq w(H_1) + w(H_2) + w(a,b).$
\end{proof}

\begin{algorithm}[hbtp]
\SetKwInOut{Input}{input}\SetKwInOut{Output}{output}
\Output{Tour $T_1$}
Let M be a greedy matching of size $k=\frac{N}{3}$, and $B$ be the nodes not in $M$\;
Complete $M$ using Lemma~\ref{lem_tourcompletion} to form a Hamiltonian path $H_T$ on $Top$\;
Form a Hamiltonian path $H_B$ on $B$ using the following randomized algorithm.\;
\textbf{Randomized Path Algorithm} \;
Form a random permutation on the nodes in B\;
Join the nodes in the same order to form the path\;
(i.e., join the first and second nodes, second and third, and so on.)\;
\textbf{Final Output}
$T_1$ is the output formed by using Lemma~\ref{lem_pathtotour} for $H_1 = H_B$ and $H_2 = H_T$.
\caption{First Subroutine of the randomized algorithm for Max TSP}
\label{alg_subroutine1}
\end{algorithm}

We now show lower bounds on the weight of $T_1$. Below $T^*$ denotes the optimum-weight tour.
\begin{lemma}
\label{lem_tsppart1}
The following is a lower bound on the weight of the tour returned by Algorithm~\ref{alg_subroutine1}
$$E[w(T_1)] \geq [\frac{3}{8} - \frac{3}{4N}]w(T^*) + \frac{6}{N}\sum_{x,y \in B}w(x,y) .$$
\end{lemma}

\begin{proof}
Using linearity of expectations, we get that,
$$E[w(T_1)] = E[w(H_T)] + E[w(H_B) + w(T_1 \setminus (H_T \cup H_B))].$$ For any given Hamiltonian path (using only the nodes in $B$) $H_B$, suppose that $T_B$ denotes the tour obtained by completing the dangling nodes to form a cycle. From Lemma~\ref{lem_pathtotour}, we get that $w(T_1 \setminus (H_T \cup H_B))] \geq w(H_B(f), w(H_B(l)).$ In other words, even though the sub-routine outputs a random tour, the weight of the edges in $T$ but not $H_T$ and $H_B$ is at least the weight of the edge between the two dangling nodes of $H_B$. This is true because we completed the tour in a very particular fashion using Lemma~\ref{lem_pathtotour}. Abusing notation, we get
$$E[w(T_1)] = E[w(H_T)] + E[w(T_B)],$$ where $E[w(T_B)]$ is the expected weight of any tour formed using the nodes only in $B$. However, using symmetry arguments, we can conclude that $E[w(T_B)] \geq \frac{|T_B|}{|B|(|B|-1)/2}\sum_{x,y \in B}w(x,y) \geq \frac{6}{N}\sum_{x,y \in B}w(x,y).$

Next, applying Lemma~\ref{lem_tourcompletion}, we get that $E[w(H_T)] = [\frac{3}{2} - \frac{3}{N}]w(M) \geq [\frac{3}{4} - \frac{3}{2N}]w(M^*) \geq [\frac{3}{8} - \frac{3}{4N}]w(T^*).$ We used the facts that $w(M) \geq \frac{w(M^*)}{2}$, and $w(M^*) \geq \frac{w(T^*)}{2}$.

Summing up the two parts gives the desired result.\hfill$\qed$
\end{proof}

Before describing the second sub-routine, we introduce the notion of an alternating-tour. Given two equal-sized disjoint sets $A,B$, we say that $T_{AB}$ is a alternating tour if it is a tour, and it alternates between nodes in $A$ and $B$. We can similarly define the notion of an alternating Hamiltonian path or just alternating path. Notice that an alternating path (or even a tour) can be represented as a sequence of alternating nodes from $A$ and $B$ respectively. In the following algorithm, we form an alternating path by adding nodes sequentially to $H_{AB}$. Finally, recall that $M$ is a greedy matching of size $k=\frac{N}{3}$, and $B$ is set of $\frac{N}{3}$ nodes not in $M$.

\begin{algorithm}[hbtp]
\SetKwInOut{Input}{input}\SetKwInOut{Output}{output}
\Output{Tour $T_2$}
Let M be a greedy matching of size $k=\frac{N}{3}$, and $B$ be the nodes not in $M$\;
Select $\frac{N}{6}$ edges uniformly at random from $M$\;
Complete these edges using Lemma~\ref{lem_tourcompletion} to form a Hamiltonian path $H_T$ with $\frac{N}{3}$ nodes\;
Let $A$ be the set of nodes in $Top$ but not in $H_T$\;
\textbf{Randomized Alternating Path Algorithm}\;
Initialize $H_{AB} = \emptyset$\;
Select one node uniformly at random from $A$\;
Select one node uniformly at random from $B$\;
Add both the nodes to $H_{AB}$ in the same order\;
Remove them from $A$ and $B$ respectively \;
Repeat the above process until $A=B=\emptyset$\;
\textbf{Final Output}\;
$T_2$ is the output formed by using Lemma~\ref{lem_pathtotour} for $H_1 = H_{AB}$ and $H_2 = H_T$.
\caption{Second Subroutine of the randomized algorithm for Max TSP}
\label{alg_subroutine2}
\end{algorithm}

\noindent\textbf{Analysis of Sub-routine $2$}\\
We begin by defining some additional notation required for the analysis of Algorithm~\ref{alg_subroutine2}. Suppose that $\mathcal{T}_2$ represents the set of all tours that are output by the algorithm with non-zero probability (support of $T_2$). Notice that for any $T' \in \mathcal{T}_2$, we can uniquely divide $T'$ into $H'_T$, $H'_{AB}$, and the connecting edges, where $H'_T$ is the path containing only the edges inside $Top \setminus A$, and $H'_{AB}$ is the alternating path between $A$ and $B$. Therefore, we have

$$E[w(T_2)] = E[w(H_T)] + E[w(H_{AB}) + w(T_2 \setminus (H_T \cup H_{AB}))].$$

Next, consider the two edges connecting $H_T$ and $H_{AB}$, i.e., $T_2 \setminus (H_T \cup H_{AB})$. Since these edges were chosen using the mechanism of Lemma~\ref{lem_pathtotour}, this means that $w(T_2 \setminus (H_T \cup H_{AB})) \geq w(H_{AB}(f), H_{AB}(l))$, for a given alternating path $H_{AB}$. And so, abusing notation, we get

$$E[w(T_2)] \geq E[w(H_T)] + E[w(T_{AB})],$$

where $E[w(T_{AB})]$ is the expected weight of any alternating tour formed using the nodes in $A,B$. Moreover, this tour is specifically formed by the sequential random mechanism described in Algorithm~\ref{alg_subroutine2} to choose a Hamiltonian path, followed by a deterministic step where the path is completed to form a cycle. In what follows, we establish a lower bound on this quantity.

\begin{lemma}
\label{lem_tsppart21}
The expected weight of the random alternating tour $E[w(T_{AB})]$ formed using the nodes in $A,B$, is at least $\frac{3}{N}\sum_{x \in Top, y \in B}w(x,y)$.
\end{lemma}

\begin{proof}
First, notice that
\begin{equation}
\label{eqn_tsptemp2}
E[w(T_{AB})] = \sum_{\substack{S \subset Top \\ |S|=\frac{N}{3}}}E[w(T_{AB})| A=S] Pr (A=S).
\end{equation}
The above equation follows from the fact that for every possible tour returned by the random algorithm (say $T_{AB} = T'_{AB}$, and $A=S$), we have that $Pr(T_{AB} = T'_{AB}) = Pr(T_{AB} = T'_{AB} | A=S) Pr(A=S)$.

Now, fix $A=S$: we have that
$$E[w(T_{AB})| A=S] = \sum_{x \in S, y \in B}w(x,y) Pr((x,y) \in T_{AB}).$$
By symmetry of our random process, we know that every edge $(x,y)$ between $S$ and $B$ has the same probability of being included into the tour. Consider the sum of such probabilities:

\begin{eqnarray*}
\sum_{x \in S, y \in B} Pr((x,y) \in T_{AB}) =\\
 \sum_{x \in S, y \in B}\sum_{tour~ T\ni(x,y)}Pr(T_{AB}=T) = \\
 \sum_{tour~ T} \frac{2N}{3}Pr(T_{AB}=T) = \frac{2N}{3}
\end{eqnarray*}

The sum above is over all tours $T$ of nodes $S\cup B$. Since there are $\frac{N}{3}$ nodes in both $S$ and $B$, then there are $\frac{N^2}{9}$ edges between them, and thus the probability of any edge $(x,y)$ from $S$ to $B$ being in the tour $T_{AB}$ is exactly $\frac{6}{N}$. Thus, we have that:

%
%
%
%

\begin{align*}
E[w(T_{AB})| A=S] & = 
 \frac{6}{N}\sum_{\substack{x \in S \\ y \in B}}w(x,y).
\end{align*}
The rest of the proof follows from rearranging Equation~\ref{eqn_tsptemp2} as a summation across the nodes in $Top$, and observing that every node from $Top$ is chosen into $A$ with probability one-half.\hfill$\qed$
\end{proof}

Our next lemma completes the lower bound for the second sub-routine.

\begin{lemma}\label{lem_tsplast}
The expected weight of the tour returned by Algorithm~\ref{alg_subroutine2} is at least $[\frac{11}{16} - \frac{3}{4N}]w(T^*)- \frac{6}{N}\sum_{x,y \in B}w(x,y)$.
\end{lemma}

\begin{proof}

Substituting the bound obtained in Lemma~\ref{lem_tsppart21} into the Equation for $E[w(T_2)]$, we get,

$$E[w(T_2)] = E(w(H_T)] + \frac{3}{N}\sum_{x \in Top, y\in B}w(x,y).$$

The rest of the proof for obtaining a lower bound on $E[w(T_2)]$ is as follows. 
Suppose that $M^*$ is the maximum weight matching on $\mathcal{N}$. By applying Lemma \ref{lem_matchingupper} to the set $T=B$ (so $n=N/3$), we know that $w(M^*) \leq \frac{6}{N}\sum_{x \in B, y \in B}w(x,y) + \frac{3}{N}\sum_{x \in Top, y \in B}w(x,y)$. Using the fact that $w(T^*) \leq 2w(M^*)$, we get $\frac{3}{N}\sum_{x \in Top, y \in B}w(x,y) \geq \frac{w(T^*)}{2} - \frac{6}{N}\sum_{x \in B, y \in B}w(x,y).$

Note that $H_T$ is obtained by randomly taking half of the matching $M$ and then completing it. From Lemma~\ref{lem_tourcompletion}, we thus have that $E[w(H_T)] \geq [\frac{3}{4} - \frac{3}{N}]w(M)$. Since by our construction, $w(M)$ is at least half of $w(M^*)$, which is at least half of $w(T^*)$, we thus have that $E[w(H_T)] \geq [\frac{3}{16} - \frac{3}{4N}]w(T^*).$


In summary, we have

$$E[w(T_2)] = (\frac{3}{16} + \frac{1}{2} - \frac{3}{4N})w(T^*) - \frac{6}{N}\sum_{x,y \in B}w(x,y).\hfill\qed$$
\end{proof}

We also know from Lemma~\ref{lem_tsppart1} that

$$E[w(T_1)] = (\frac{3}{8} - \frac{3}{4N})w(T^*) + \frac{6}{N}\sum_{x,y \in B}w(x,y).$$

The final bound is obtained by using $E[w(T)] = \frac{1}{2}(E(w(T_1)] + E[w(T_2)])$.

\section{Lower Bounds}
\subsubsection*{Densest $k$-Subgraph}
\begin{claim}
No ordinal approximation algorithm, deterministic or randomized, can provide an approximation factor better than $2$ for Densest $k$-subgraph.
\end{claim}
\begin{proof}
Since randomized algorithms are more general than deterministic algorithms, it suffices to show the claim just for Randomized Algorithms.

Given a parameter $k$, consider an instance of the Densest $k$-Subgraph with $Mk$ nodes for a large enough value of $M$ (say $M$ is much larger than $k$). The set of nodes in the graph can be divided into $M$ clusters $N_1, N_2, \ldots, N_M$, each containing $k$ nodes. The preference ordering is given as follows: for a given $i$, every node in $N_i$ prefers all the nodes in $N_i$ over every node outside of $N_i$. The exact preference ordering within $N_i$ and outside of $N_i$ can be arbitrary.

Now, randomly choose one of $M$ clusters and assign a weight of $2$ to all the edges strictly inside that cluster. Assign a weight of $1$ to every other edge in the graph. It is easy to see that these weights induce the given preference orderings. Now, without loss of generality, it suffices to consider only algorithms that choose $k$ nodes within a fixed cluster. Moreover, since the clusters are identical from the ordinal point of view, the optimum algorithm for this instance just picks one of the $M$ clusters uniformly at random, and therefore, its approximation ratio is $\frac{2}{1+\frac{1}{M}}$ which approaches $2$ as $M \to \infty$
\end{proof}

\end{document}